\documentclass[11pt]{article}
\usepackage{amssymb, authblk}
\usepackage{amsmath,bbm}
\usepackage{fullpage}
\usepackage{amsthm, hyperref}
\usepackage{graphicx}
\usepackage{verbatim}
\usepackage{algorithm}
\usepackage{algpseudocode}
\usepackage[dvipsnames]{xcolor}
\usepackage{tikz}
\usetikzlibrary{arrows.meta}
\usetikzlibrary{decorations.pathreplacing}
\DeclareMathAlphabet{\mathpzc}{OT1}{pzc}{m}{it}
\usepackage{natbib}
\usepackage{bigints}

\algnewcommand\algorithmicinput{\textbf{INPUT:}}
\algnewcommand\INPUT{\item[\algorithmicinput]}
\algnewcommand\algorithmicoutput{\textbf{OUTPUT:}}
\algnewcommand\OUTPUT{\item[\algorithmicoutput]}

\usepackage{hyperref}[]
\hypersetup{
    colorlinks=true,
    linkcolor=blue,
    filecolor=magenta,      
    urlcolor=cyan,
      citecolor=blue,
    }

\newtheorem{theorem}{Theorem}

\newtheorem{lemma}[theorem]{Lemma}
\newtheorem{proposition}[theorem]{Proposition}

\newtheorem{remark}{Remark}
\newtheorem{assumption}{Assumption}

\allowdisplaybreaks

\DeclareMathOperator*{\argmax}{arg\,max}

\newcommand{\dint}{\,\mathrm{d}}

\DeclareFontFamily{U}{mathx}{\hyphenchar\font45}
\DeclareFontShape{U}{mathx}{m}{n}{
	<5> <6> <7> <8> <9> <10>
	<10.95> <12> <14.4> <17.28> <20.74> <24.88>
	mathx10
}{}
\DeclareSymbolFont{mathx}{U}{mathx}{m}{n}
\DeclareFontSubstitution{U}{mathx}{m}{n}
\DeclareMathAccent{\widecheck}{0}{mathx}{"71}
\DeclareMathAccent{\wideparen}{0}{mathx}{"75}

\title{Dynamic and heterogeneous treatment effects with abrupt changes}
\author[1]{Oscar Hernan Madrid Padilla}
\author[2]{Yi Yu}
\affil[1]{Department of Statistics, University California, Los Angeles}
\affil[2]{Department of Statistics, University of Warwick}

\begin{document}

\maketitle

\begin{abstract}
From personalised medicine to targeted advertising, it is an inherent task to provide a sequence of decisions with historical covariates and outcome data.  This requires understanding of both the dynamics and heterogeneity of treatment effects.  In this paper, we are concerned with detecting abrupt changes in the treatment effects in terms of the conditional average treatment effect (CATE) in a sequential fashion.  To be more specific, at each time point, we consider a nonparametric model to allow for maximal flexibility and robustness.  Along the time, we allow for temporal dependence on historical covariates and noise functions.  We provide a  kernel-based change point estimator, which is shown to be consistent in terms of its detection delay, under an average run length control.  Numerical results are provided to support our theoretical findings. 
\end{abstract}

\section{Introduction}\label{sec-introduction}

In a targeted marketing example, one may wish to be able to decide whether or not to preach to a potential customer based on their historical data, including their purchase history on related products, their lifestyle, and their reactions to previous advertisements, among many others.  In order to maximise the net profit, this decision is ideally made to maximise a certain loss function based on the potential outcome of launching advertisements to this potential customer.  Over time, it is unrealistic to assume that a customer reacts to advertisements similarly, even if all the relevant covariates stay unchanged.  It is therefore beneficial to detect when the effects of advertisements abruptly change, so that companies can  deliver new marketing strategies to maintain profits.  In fact, not only in marketing, understanding the dynamics and heterogeneity of treatments' causal effects is vital in many application areas, including medicine studies on the impact  of certain drugs on health indices \citep[e.g.][]{choi2017dynamic, mumford2018circulating} and political sciences studies on the effects of certain government policies \citep[e.g.][]{blackwell2013framework, park2017temporal}, among many others.

In response to such demand, we have witnessed a surge of statistical causal inference methods and theory in recent years, built upon the potential outcome framework \citep{neyman1923applications, rubin1974estimating}, where a treatment effect is defined as the difference between two potential outcomes.  With access to covariates, in addition to the treatment assignment and outcomes, it is natural to model treatment effects as functions of covariates, i.e.~to understand the heterogeneity of the treatment effects.  Statistical efforts on this front include methods based on random forests \citep[e.g.][]{wager2018estimation, green2012modeling}, weighted linear regression \citep[e.g.][]{sun2021estimating}, Lasso-type regularisation \citep[e.g.][]{nie2021quasi} and nearest-neighbour estimation \citep[e.g.][]{gao2020minimax}, to name but a few.

On top of the heterogeneity in terms of the covariates, in a dynamic treatment regime, where for each subject, a sequence of decisions on whether to impose a treatment is to be made.  This can be understood as heterogeneity of treatment effects time-wise.  Statistical research on understanding the dynamics often focuses on the average effects, started from \cite{robins1986new}, which aims to solve a g-formula involving multiple integrals.  Such formulation is widely adopted in more recent work \citep[e.g.][]{murphy2003optimal} and has fostered a number of variants, including penalisation-type estimators \citep[e.g.][]{chakraborty2010inference} and others.  The g-formula, a chain conditional expectation based method, considers the temporal dependence.  The multiple integral, however, hinders analysis with moderately many time points.  This, therefore, motivates us to take a time series perspective instead.

In this paper, we are concerned with data
    \begin{equation}\label{eq-data}
        \{(Y_{t,i}, X_{t,i}, Z_{t,i}), \, i = 1, \ldots, n, \, t \in  N^*\} \subset R \times R^p \times \{0, 1\},
    \end{equation}
    where $Y_{t, i}$ represents the outcome of subject $i$ at time point/stage $t$, $X_{t, i}$ denotes the corresponding covariates and $Z_{t, i}$ is the treatment indicator.  We adopt the potential outcome framework \citep{neyman1923applications, rubin1974estimating}, assuming that
    \begin{equation}\label{eq-model-1}
        Y_{t, i} = Y_{t, i}(1) Z_{t, i} + Y_{t,i}(0)(1 - Z_{t,i}),
    \end{equation}
    with
    \begin{equation}\label{eq-model-2}
        Y_{t, i}(z) = \mu_{t, z}(X_{t, i}) + \epsilon_{t, i}(z), \quad z \in \{0, 1\},
    \end{equation}
    where $\mu_{t, 0}, \mu_{t, 1}:\, \mathcal{X} \to  R$ are mean functions.  We allow for temporal dependence in the multivariate sequence $\{(\epsilon_{t, i}(0), \epsilon_{t,i}(1), X_{t,i}), \, t \in N^*\}$, detailed in Assumption \ref{assume-model-1}.
    
In a dynamic treatment regime, it is crucial to deal with the ever-changing nature.  In this paper, we are especially interested in detecting abrupt changes in treatment effects.  Such a task falls into the area of change point detection, dating back to \cite{wald1945sequential} and recently undergoing a renaissance, becoming the host of a large volume of literature.  At a high level, depending on the availability of the data, change point analysis can be categorised into online/sequential change point analysis - where the data are coming in while one is monitoring if a change point has just occurred, and offline change point analysis - where one detects change points retrospectively.  Motivated by the sequential advertisement example we mentioned at the beginning, we are dealing with an online change point detection problem  in this paper.  In Section \ref{sec-problem-formulation}, we provide formal descriptions of the tasks and relevant literature.

We summarise the contributions of this paper below.  Firstly, to the best of our knowledge, this is the first attempt  of understanding the heterogeneity and dynamics of treatment effects, in the presence of change points.  We estimate the treatment effect in a nonparametric model, which amplifies the model flexibility.  This set of analysis on its own is interesting due to the lack of systematic nonparametric online change point analysis results.  Secondly, to avoid dealing with multiple integrals and to preserve the dependence on the history in a dynamic treatment regime, we allow for temporal dependence across time on both the covariate and noise sequences while controlling the false alarms in detecting change points.  

\subsection*{Notation}  

For any vector $v$, let $\|v\|$ and $\|v\|_{\infty}$ be its $\ell_2$- and entrywise supreme norms.  For any function $f(\cdot)$, $\|f\|_{\infty} = \sup_{x \in \mathrm{Dom}(f)} |f(x)|$.  Let $N$ and $N^*$ denote the collection of all natural numbers and all positive natural numbers, respectively.

\section{Problem formulation}\label{sec-problem-formulation}

There are three key ingredients of our problem: (1) nonparametric and temporal-dependence analysis; (2) causal inference; and (3) an online change point detection framework.  The requirements for each of these are gathered in Assumptions~\ref{assume-model-1}, \ref{assump-causal} and \ref{assump-change-point-scenarios}, respectively.

\begin{assumption}[Model] \label{assume-model-1}
For data specified in \eqref{eq-data}, \eqref{eq-model-1} and \eqref{eq-model-2}, we assume the followng.
    
\textbf{(a.)} (Dependence.) Assume that for any $i \neq j$, $i, j \in \{1, \ldots, n\}$,
        \[
            \{(\epsilon_{t,i}(0),\epsilon_{t,i}(1), X_{t,i}), \, t \in  N^*\} \perp \{(\epsilon_{t,j}(0),\epsilon_{t,j}(1), X_{t,j}), \, t \in N^*\}.
        \]
        For any $i \in \{1, \ldots, n\}$, assume that the multivariate sequence $\{(\epsilon_{t, i}(0), \epsilon_{t,i}(1), X_{t,i}), \, t \in N^*\}$ is strictly stationary and $\alpha$-mixing with coefficients $\{\alpha_m\}_{m \in  N^*}$ satisfying $\alpha_m \leq \exp(-C_{\alpha} m^{\gamma_{\alpha}})$, where $C_{\alpha}, \gamma_{\alpha} > 0$ are absolute constants.
        
\textbf{(b.)} (Covariates.)  Assume that the support of the covariates $\mathcal{X} \subset  R^p$ is a compact set.  In addition, assume that $\{X_{t, i}\}_{t \in N^*, i = 1, \ldots, n}$ are identically distributed with marginal density function $g(\cdot)$ satisfying that
	\[
	    c_{g, 1} \leq \inf_{x \in \mathcal{X}} g(x) \leq \sup_{x \in \mathcal{X}} g(x) \leq c_{g, 2},
	\]
	where $c_{g, 1}, c_{g, 2} > 0$ are absolute constants.  Assume that $g(\cdot)$ is Lipschitz continuous with Lipschitz constant $C_{\mathrm{Lip}} > 0$.
	
\textbf{(c.)} (Time varying conditional average treatment effects.)  Let the conditional average treatment effect (CATE) at time $t$, conditional on $X_t = x \in \mathcal{X}$, be
    \[
        \tau_t(x) =  E\{Y_{t, 1}(1) - Y_{t, 1}(0) | X_{t, 1}  = x\}.
    \]
    For any $t \in  N^*$, assume that $\tau_t(\cdot)$ is Lipschitz continuous with Lipschitz constant $C_{\mathrm{Lip}} > 0$.
	
\textbf{(d.)} (Error functions.)  Assume that $E\{\epsilon_{t, i}(z)| x\}$ is a mean zero sub-Gaussian random variable, with sub-Gaussian parameter $\sigma$, for any $x \in \mathcal{X}$, $t \in  N^*$, $i \in \{1, \ldots, n\}$ and $z \in \{0, 1\}$.  
\end{assumption}

To incorporate temporal dependence, we assume the multivariate sequence $\{(\epsilon_{t, i}(0), \epsilon_{t,i}(1), X_{t,i})$, $t \in  N^*\}$ to be $\alpha$-mixing, for each subject $i \in \{1, \ldots, n\}$.  We refer the readers to the survey \cite{bradley2005basic} for detailed definitions and discussions.  In particular, we require the dependence decay rate $\alpha_m$ to be exponential.  In modern time series literature, different forms of temporal dependence are imposed, including different mixing conditions and functional dependence \citep{wu2005nonlinear}.  To our best of knowledge, theoretical guarantees for non-univariate data are all developed based on temporal dependence with exponential decay \citep[e.g.][]{wong2020lasso, dedecker2004coupling, merlevede2011bernstein}, which corresponds to a short range dependence. 

For model robustness and flexibility, we consider a nonparametric framework.  As for the distribution of the covariates, detailed in Assumption \ref{assume-model-1}\textbf{(b.)}, we require that the density exists with standard regularity conditions.  For the mean functions, as a direct consequence of $\mathcal{X}$ being compact, the mean functions $\mu_{t, z}(\cdot)$ are bounded.  Since we are interested in understanding the treatment effect, we do not impose any further conditions on the mean functions \emph{per se}.  We, instead, collect all the relevant conditions in Assumption \ref{assume-model-1}\textbf{(c.)}~on CATE, which is a powerful tool in learning the heterogeneity of treatment effects \citep[e.g.][]{abrevaya2015estimating, jacob2021cate}.  We assume the CATE functions at every time point to be Lipschitz.  This allows the individual mean functions to be of any form provided that their difference functions are manageable. To illustrate this, we create rather rough mean functions but with smooth difference $\tau(\cdot)$ (see Section \ref{sec:ad_fig} for details).  We depict the estimators based on estimating two mean functions (Two-K) and estimating $\tau(\cdot)$ (One-K) directly in Fig~\ref{fig4}, where we can see the advantage of estimating the CATE function directly versus estimating individually the mean functions for both treatment and control groups and then taking their difference.

\begin{figure}[ht!]
	\begin{center}
		\includegraphics[width = 0.4\textwidth]{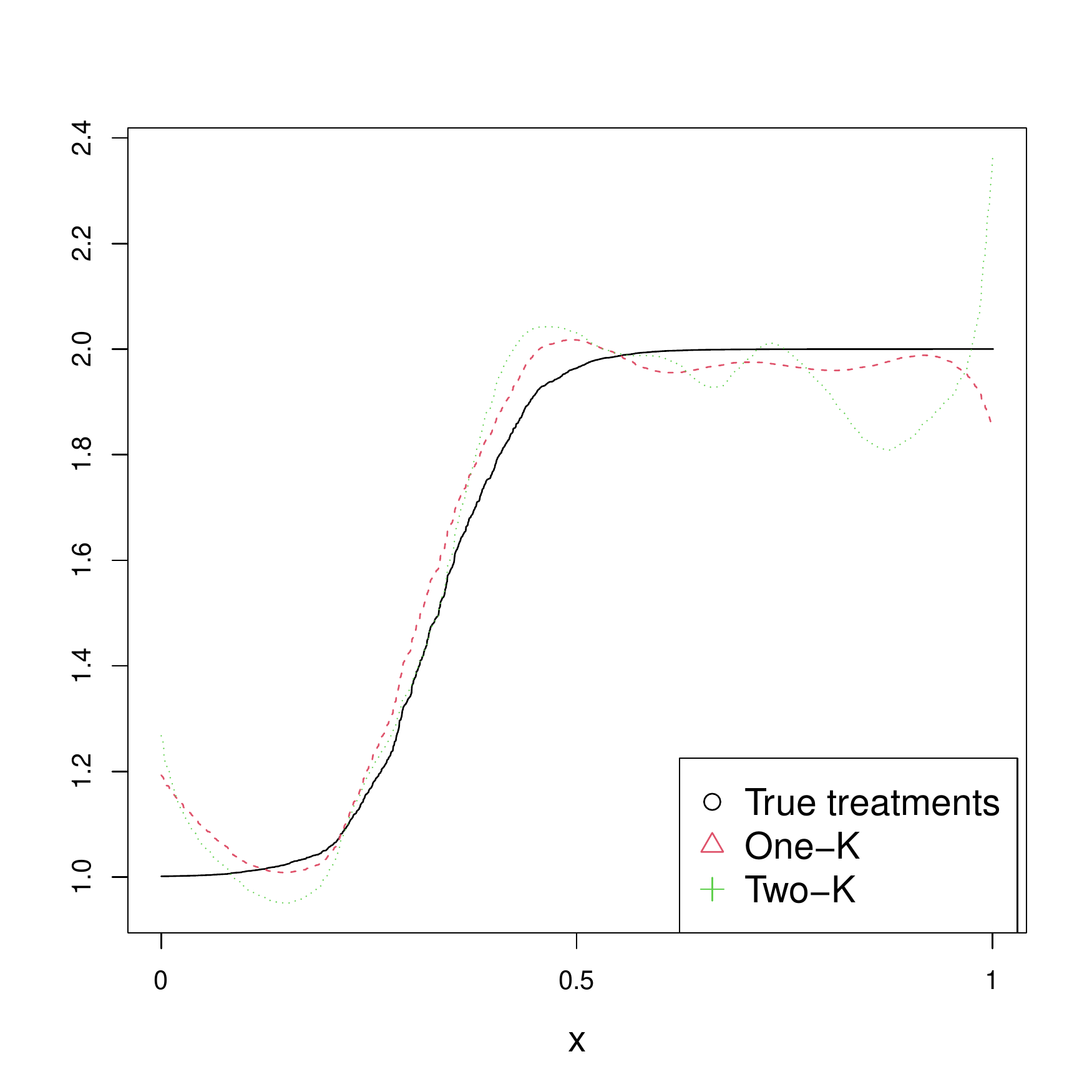} 
		\caption{\label{fig4} Different methods estimating $\tau(\cdot)$.}
	\end{center}
\end{figure}

The conditions on noise functions are considered in Assumption \ref{assume-model-1}\textbf{(d.)}, where we allow the additive noise to be functions of the covariates, provided the conditional means are zero.  Such dependence allows for heteroscedasticity.

To ensure the causal inference problem we are interested in is a valid question, we shall make sure the quantity of interest is identifiable.  We therefore impose the following three standard assumptions.  Assumption \ref{assump-causal}\textbf{(a.)}, \textbf{(b.)}~and \textbf{(c.)}~are often referred to as Neyman--Rubin potential outcome model, unconfoundedness and overlap.  These three conditions together identify the CATE in observational studies \citep[e.g.][]{imbens2015causal}.
\begin{assumption}[Causal inference]\label{assump-causal}
For any $z \in \{0, 1\}$, $t \in N^*$ and $i \in \{1, \ldots, n\}$, assume that: \textbf{(a.)} $Y_{t, i} = Y_{t, i}(z)$, if $Z_{t, i} = z$; \textbf{(b.)} $\{Y_{t, i}(0), Y_{t, i}(1)\} \perp Z_{t, i} | X_{t, i}$; \textbf{(c.)} there exists an absolute constant $c_{\mathrm{prop}} > 0$ such that
    \[
        \inf_{x \in \mathcal{X}}\min\left\{\pi(x), \, 1 - \pi(x)\right\} \geq c_{\mathrm{prop}},
    \]
    where $\pi(x) = \mathbb{P}\{Z_{t, i} = 1 | X_{t, i} = x\}$ is called the propensity score.
\end{assumption}

For simplicity, we assume in Assumption \ref{assump-causal}\textbf{(c.)}~that the propensity score is invariant of time.  When the true propensity scores are known, e.g.~in a randomised design case, the rest of this paper can be adapted to varying propensity scores straightforwardly.   

Our main task is to detect change points in the CATE functions.  The change point scenarios are detailed below.

\begin{assumption}[Change point scenarios]\label{assump-change-point-scenarios}
\

\textbf{(a.)~No change point scenario.} Assume that $\tau_1 = \tau_2 = \cdots$. 

\textbf{(b.)~One change point scenario.} Assume that there exists $\Delta \in N^*$ such that $\tau_1 = \cdots = \tau_{\Delta} \neq \tau_{\Delta + 1} = \cdots$.  Let the jump size be $\kappa = \|\tau_{\Delta} - \tau_{\Delta+1}\|_{\infty}$.
\end{assumption}

Assumption \ref{assump-change-point-scenarios}\textbf{(a.)}~formalises the no change point scenario.  Assumption \ref{assump-change-point-scenarios}\textbf{(b.)}~specifies that when a change point occurs, we characterise the change in terms of the function supreme norm of the difference between two different CATE functions.  We let $\Delta$ denote the pre-change sample size.  

In an online change point detection problem, the goal is twofold, to  detect a change point as soon as it occurs, while controlling false alarms.  To be specific, for a pre-specified $\Gamma > 0$, if there is no change point, one wishes to show the change point estimator $\widehat{\Delta}$ satisfies that
\[
    E_{\infty}\{\widehat{\Delta}\} \geq \Gamma,
\]
where the subscript $\infty$ of $E$ denotes the no change point scenario, i.e.~$\Delta = \infty$.  The quantity $E_{\infty}\{\widehat{\Delta}\}$ is referred to as the average run length (ARL), a lower bound on which is a popular method in controlling false alarms.  This guarantees that, when there is no change point, the expectation of the false alarm location is at least $\Gamma$. 

If there exists a change point $\Delta < \infty$, then one wishes to show that, with probability tending to one as the pre-change sample size $\Delta$ diverges, it holds that 
\[
    \widehat{\Delta} < \Delta + \mathrm{delay},
\]
where delay denotes the detection delay.  We call a change point estimator is consistent if $\mathrm{delay}/\Delta \to 0$, as $\Delta$ grows unbounded.  Recent literature guaranteeing consistent detection delay under false alarm controls includes \cite{lai1995sequential}, \cite{yu2020note}, \cite{berrett2021locally}, \cite{keshavarz2020sequential} among many others.  

\section{Online CATE change point detector}\label{sec-algorithm}

Due to Assumption \ref{assump-causal}\textbf{(b.)}~and \textbf{(c.)}, as demonstrated in Lemma \ref{lem-justifying-estimator}, it holds that 
    \begin{equation}\label{eq-motivation}
        E\left\{Y_{t, i} \left(\frac{Z_{t, i}}{\pi(x)} - \frac{1 - Z_{t, i}}{1 - \pi(x)}\right) \Big| X_{t, i} = x\right\} = \tau_t(x), \quad t \in  N^*.
    \end{equation}
    We therefore consider a change point estimator $\widehat{\Delta}$, based on a kernel regression estimator, which is a function of $Y_{t, i} [Z_{t, i}/\pi(x) - \{1 - Z_{t, i}\}/\{1 - \pi(x)\}]$.  The detailed algorithm is collected in Algorithm \ref{alg-main}.  Note that we directly form estimators of the CATE function, rather than estimating the two mean functions on the treated and controlled groups, separately.  This method has two benefits: (a) since one less function is estimated, fewer tuning parameters are required in forming kernel regression estimators; and (b) the mean functions are allowed to be arbitrary, provided their difference function is well behaved.
    
\begin{algorithm}[ht]
    \begin{algorithmic}
        \INPUT $\{(Y_{t, i}, X_{t, i}, Z_{t, i})\}_{t \in  N^*, i = 1, \ldots, n}$, $w$, $h$, $\widehat{\pi}(\cdot)$, $\varepsilon$
        \State $t \leftarrow 2w$
        \State $\mathrm{FLAG} \leftarrow 0$
        \While{$\mathrm{FLAG} = 0$}
            \State $t \leftarrow t + 1$
            \State $\widehat{\tau}_{t-2w, t-w}(\cdot) \leftarrow \frac{\sum_{l = t-2w+1}^{t-w} \sum_{i = 1}^n \left\{Y_{l, i} \left(\frac{Z_{l, i}}{\widehat{\pi}(X_{l, i})} - \frac{1 - Z_{l, i}}{1 - \widehat{\pi}(X_{l, i})}\right)\right\} \mathpzc{k}\left(\frac{X_{l, i} - \cdot}{h}\right)}{\sum_{l = t-2w+1}^{t-w} \sum_{i = 1}^n \mathpzc{k}\left(\frac{X_{l, i} - \cdot}{h}\right)}$
            \State $\widehat{\tau}_{t-w, t}(\cdot) \leftarrow \frac{\sum_{l = t-w+1}^t \sum_{i = 1}^n \left\{Y_{l, i} \left(\frac{Z_{l, i}}{\widehat{\pi}(X_{l, i})} - \frac{1 - Z_{l, i}}{1 - \widehat{\pi}(X_{l, i})}\right)\right\} \mathpzc{k}\left(\frac{X_{l, i} - \cdot}{h}\right)}{\sum_{l =t-w+1}^t \sum_{i = 1}^n \mathpzc{k}\left(\frac{X_{l, i} - \cdot}{h}\right)}$
            \State $\mathrm{FLAG} \leftarrow 1\left\{|\widehat{\tau}_{t-2w, t-w}(X_{r, i}) - \widehat{\tau}_{t-w, t}(X_{r, i})| \geq \varepsilon \right\}$
        \EndWhile
        \OUTPUT $t$
        \caption{Online CATE change point detector} \label{alg-main}
    \end{algorithmic}
\end{algorithm}    

\begin{algorithm}[ht]
	\begin{algorithmic}
	    \INPUT $\{(X_{t, i}, Z_{t, i})\}_{t = 1, i = 1}^{T, n}$,  $x$.
	    \State $\widehat{\beta} \leftarrow \argmax_{\beta \in R^p} \left[\sum_{t = 1}^T \sum_{i = 1}^n  Z_{t, i} X^{\top}_{t, i} \beta - \log\left\{1 + \exp \left(X^{\top}_{t, i} \beta\right)\right\}\right]$
		\State $\widehat{\pi}(x) \leftarrow \frac{\exp\left(x^{\top} \widehat{\beta}\right)}{1 + \exp\left(x^{\top} \widehat{\beta}\right)}$
		\OUTPUT $\widehat{\pi}(x)$
		\caption{Estimating the propensity score.} \label{alg-prop}
	\end{algorithmic}
\end{algorithm}

In Algorithm \ref{alg-main}, we search a potential change point using a sliding window argument, based on a propensity score estimator $\widehat{\pi}(\cdot)$.  In Algorithm \ref{alg-prop}, for completeness, we provide an estimator obtained from a logistic regression.  This is previously used in  \cite{ye2021non}.  For Algorithm \ref{alg-main} \emph{per se}, any propensity score estimator $\widehat{\pi}(\cdot)$ can be used as an input.  In practice, when the data are acquired from experiments, the true propensity scores are known and can be used as the input here.  For the theoretical guarantee of Algorithm \ref{alg-main}, we will detail the requirements of $\widehat{\pi}(\cdot)$ in Theorem \ref{thm-main}.  As for the computational  complexity of Algorithm \ref{alg-main}, denoting the computational cost of computing $\hat{\pi}$ as $O(C_{\pi})$, the total cost per iteration of   Algorithm \ref{alg-main} is $O(w n d +  C_{\pi})$.

To estimate CATE functions $\tau(\cdot)$ in a nonparametric model, we construct  Nadaraya--Watson estimators \citep[e.g.~Definition 5.39 in][]{wasserman2006all} $\widehat{\tau}_{t_1, t_2}(\cdot)$ following \eqref{eq-motivation}.  To detect changes,  after acquiring a new data point at time $t$, we compare $\widehat{\tau}_{t-2w, t-w}(\cdot)$ and $\widehat{\tau}_{t-w, t}(\cdot)$.  In Assumption \ref{assump-change-point-scenarios}\textbf{(b.)}, we measure the jump using the supreme norm, which guides us to calculate the sample version of the supreme norm of $\widehat{\tau}_{t-2w, t-w}(\cdot) - \widehat{\tau}_{t-w, t}(\cdot)$, given the available data.  

Different types of windows are commonly used in change point detection literature \citep[e.g.][]{fryzlewicz2014wild, kovacs2020seeded, eichinger2018mosum, li2015m, aminikhanghahi2017survey}.  In Algorithm \ref{alg-main}, we adopt a sliding window routine with a fixed window width, denoted by $w$.  In view of different windows deployed in the literature, window widths are chosen in various ways, including all possible intervals \citep[e.g.][]{yu2020note}, random intervals \citep[e.g.][]{fryzlewicz2014wild}, a dyadic grid of intervals \citep[e.g.][]{kovacs2020seeded}, and others.  We would like to highlight that not all the aforementioned methods have been adopted in an online change point detection setup.  Using a fixed window width improves the computational efficiency, but undoubtedly the performance of the change point estimator depends on the choice of the window width.  We leave the theoretical and practical discussions on the tuning parameters to Sections~\ref{sec-theory} and \ref{sec-numerical}, respectively.

\begin{remark}
Another popular sliding window technique used in online change point detection \citep[see e.g.,][]{aminikhanghahi2017survey,madrid2019sequential} is as follows.  For a pre-specified window width $W$, when acquiring a new data point at time $t$, one checks every $s \in \{t-W + 1, \ldots, t-1\}$ for potential change points, by evaluating the difference between estimators built on $\{t-W, \ldots, s\}$ and $\{s + 1, \ldots, t\}$.  This strategy needs modifications for nonparametric estimation problems.  To guarantee the performance of kernel estimators, one can only evaluate $s \in \{t- W + w, \ldots, t-w\}$, with a further tuning parameter $w$.  Compared to the fixed-width sliding window used in Algorithm \ref{alg-main}, this alternative method involves more tuning parameters, has a higher computational cost, but has the same bottleneck that is the detection delay being at least $w$.  Based on this discussion, we stick to the fixed-width window in this paper.
\end{remark}

\section{Consistent change point detection under ARL control}\label{sec-theory}

The Online CATE change point detector (Algorithm \ref{alg-main}) is constructed based on Nadaraya--Watson estimators.  Before presenting the theoretical guarantees of Algorithm \ref{alg-main}, Assumption \ref{assump-kernel} collects assumptions on the kernel function.

\begin{assumption}[The kernel function] \label{assump-kernel}
Assume that there exists an absolute constant $C_{\mathpzc{k}} > 0$ such that 
    \[
        \max \left\{\sup_{u \in R^d} |\mathpzc{k}(u)|, \, \int_{R^d} |\mathpzc{k}(u)|\dint u, \, \int_{R^d} |\mathpzc{k}(u)| \|u\| \dint u \right\} \leq C_{\mathpzc{k}}.
    \]

Assume that there exists an absolute constant $L > 0$ such that at least one of the following holds.  \textbf{(a.)} The kernel function $\mathpzc{k}(\cdot)$ is Lipschitz continuous with Lipschitz constant $C_{\mathrm{Lip}} > 0$.  For any $u$ satisfying $\|u\|_{\infty} > L$, $\mathpzc{k}(u) = 0$.  \textbf{(b.)} The kernel function $\mathpzc{k}(\cdot)$ is continuous with its derivative satisfying $\|\mathpzc{k}'\|_{\infty} \leq C_{\mathrm{Lip}}$.  There exists an absolute constant $v > 1$ such that for any $u$ with $\|u\|_{\infty} > L$, $|\mathpzc{k}'(u)| \leq C_{\mathrm{Lip}} \|u\|_{\infty}^{-v}$.
\end{assumption}

Assumption \ref{assump-kernel} imposes boundedness and smoothness conditions on the kernel functions.  Different smoothness and tail conditions have been imposed on kernel functions in the literature \citep[see e.g.][]{gine1999laws, gine2001consistency, sriperumbudur2012consistency, hansen2008uniform}.  Following the same conditions as imposed in \cite{hansen2008uniform}, Assumption \ref{assump-kernel}\textbf{(a.)} and \textbf{(b.)} consider two different scenarios, covering a wide range of kernel functions.  

\begin{assumption}[Signal-to-noise ratio condition] \label{assump-snr}
For $\Gamma > 0$, assume that
\begin{equation}\label{eq-snr-cond}
    \kappa^{d+2} \gtrsim \frac{\sigma^2}{n} \left\{\frac{\log(\Gamma \vee \Delta)}{\Delta} \vee \frac{\log^{2/\gamma_1}(\Gamma \vee \Delta)}{\Delta^2}\right\} a_{\Delta},
\end{equation}
where $\gamma_1 = 2\gamma_{\alpha}/(2 + \gamma_{\alpha})$ and $a_{\Delta} > 0$ is any arbitrarily diverging sequence satisfying $a_{\Delta} \to \infty$, as $\Delta \to \infty$. 
\end{assumption}

Assumption \ref{assump-snr} details the signal-to-noise ratio condition in the presence of a change point, i.e.~under Assumption \ref{assump-change-point-scenarios}\textbf{(b.)}.  Recall that $\gamma_{\alpha}$ encodes the dependence, which enters the condition via $\gamma_1$.  The parameter $\gamma_1 < 2$ and $\gamma_1 \to 2$ when $\gamma_{\alpha} \to \infty$, which degenerates the independent case.  The right-hand side of \eqref{eq-snr-cond} consists of two terms, where the first term corresponds with the optimal signal-to-noise ratio condition when the data are assumed to be independent \citep{padilla2021optimal}.  Note that, when $\gamma_{\alpha} \to \infty$, i.e.~the data become independent, then we recover the optimal condition.  It is interesting to see, as long as $\Delta$ is large enough compared to $\log(\Gamma)$, it is always the first term dominates the right-hand side of \eqref{eq-snr-cond}.  This coincides with the short-range dependence assumption implied by the exponentially-decay $\alpha$-mixing condition.

We would like to mark that \cite{padilla2021optimal} deals with an offline change point detection case, while ours focuses on an online change point case.  Generally speaking, the detection delay quantity in an online problem can be regarded as the counterpart of the change point estimation error in an offline problem, but online is a more challenging case than its offline counterpart \citep[e.g.][]{yu2020review}.

\begin{theorem}\label{thm-main}
Let the data be 
    \[
        \{(Y_{t,i}, X_{t,i}, Z_{t,i}), \, i = 1, \ldots, n, \, t \in N^*\} \subset  R \times  R^p \times \{0, 1\},
    \]
    where $Y_{t, i}, X_{t, i}$ and $Z_{t, i}$ are the outcome, covariates and treatment indicator of subject $i$ at time point $t$.  Assume Assumptions~\ref{assume-model-1} and \ref{assump-causal} hold.  For any $\Gamma > 0$, let $\widehat{\Delta}$ be the output of Algorithm \ref{alg-main} with the kernel function $\mathpzc{k}(\cdot)$ satisfying Assumption \ref{assump-kernel}, the propensity score $\widehat{\pi}(\cdot)$ independent of the data.  The following $C_h, C_{\varepsilon}, c_1, c_h > 0$ are all absolute constants.
    
\medskip
\noindent \textbf{Case 1. (No change point scenario)}  If Assumption \ref{assump-change-point-scenarios}\textbf{(a.)}~holds, then for any window width $w > 0$, the kernel bandwidth
    \[
        h = C_h \left[\left\{\frac{\sigma^2 \log^{2/\gamma_1} (\Gamma w)}{nw^2}\right\}^{1/(d+2)} \vee \left\{\frac{\sigma^2 \log(\Gamma w)}{nw}\right\}^{1/(d+2)}\right]
    \]
    and the threshold 
    \[
        \varepsilon = C_{\varepsilon} \left[\left\{\frac{\sigma^2 \log^{2/\gamma_1} (\Gamma w)}{nw^2}\right\}^{1/(d+2)} \vee \left(\frac{\sigma^2 \log(\Gamma w)}{nw}\right)^{1/(d+2)} \vee \|\widehat{\pi} - \pi\|_{\infty}\right], 
    \]
    it holds that 
    \[
        E_{\infty} (\widehat{\Delta}) \geq \Gamma.
    \]
    
\medskip
\noindent \textbf{Case 2. (One change point scenario)}  If Assumption \ref{assump-change-point-scenarios}\textbf{(b.)}~holds and in addition assuming that Assumption \ref{assump-snr} holds, then for
    \begin{itemize}
        \item window width satisfying
        \[
            w \geq c_1 \left[\frac{\sigma^2 \log(\Gamma \vee \Delta)}{n \kappa^{d+2}} \vee \left\{\frac{\sigma^2 \log^{2/\gamma_1}(\Gamma \vee \Delta)}{n \kappa^{d+2}}\right\}^{1/2}\right],
        \]
        \item kernel bandwidth satisfying 
        \[
            C_h \left[\left\{\frac{\sigma^2 \log^{2/\gamma_1} (\Gamma \Delta w)}{nw^2}\right\}^{1/(d+2)} \vee \left\{\frac{\sigma^2 \log(\Gamma \Delta w)}{nw}\right\}^{1/(d+2)}\right] \leq h \leq c_h \kappa,
        \]
        \item the threshold 
        \[
            \varepsilon = C_{\varepsilon} \left(h \vee \|\widehat{\pi} - \pi\|_{\infty}\right), 
        \]
    \end{itemize}
    it holds that with an absolute constant $c > 0$,
        \[
            \mathbb{P}\left(\widehat{\Delta} \leq \Delta + w\right) \geq 1 - \Delta^{-c}.
        \]
\end{theorem}

Theorem \ref{thm-main} demonstrates the theoretical properties of the change point estimator.  As we have mentioned in Section \ref{sec-algorithm}, we allow for any propensity score estimator $\widehat{\pi}(\cdot)$, as long as it is independent of the data.  To allow for this generality, the threshold $\varepsilon$ is a function of the estimation error $\|\widehat{\pi} - \pi\|_{\infty}$.  In the rest of the discussion, for simplicity, we consider the case that 
\[
    \|\widehat{\pi} - \pi\|_{\infty} \lesssim \left\{\frac{\sigma^2 \log^{2/\gamma_1} (\Gamma w \Delta)}{nw^2}\right\}^{1/(d+2)} \vee \left\{\frac{\sigma^2 \log(\Gamma w \Delta)}{nw}\right\}^{1/(d+2)}.
\]

To lower bound the ARL, we allow for any window width $w > 0$, with properly chosen kernel bandwidth and threshold, both of which are decreasing functions of $w$.  Since the ARL considers the case when there is no change point, it is worth comparing the bandwidth choice with the standard kernel estimation literature.  We see that when $\gamma_1 = 2$, which corresponds to the independent case, the bandwidth corresponds to the optimal bandwidth in kernel estimation literature \citep[e.g.][]{Tsybakov2009}.

In the presence of a change point, we are to consider the detection delay under the signal-to-noise ratio condition Assumption \ref{assump-snr}.  Different from the no change point scenario in \textbf{Case 1}, we require a lower bound on the window width $w$.  This condition guarantees that the condition on the kernel bandwidth is not an empty set.  With the threshold reflecting the bandwidth choice, we show that the detection delay is upper bounded by the window width $w$.  

In particular, in \textbf{Case 2.}, if we choose the window width to be
\begin{equation}\label{eq-optimal-delay}
    w \asymp \frac{\sigma^2 \log(\Gamma \vee \Delta)}{n \kappa^{d+2}} \vee \left\{\frac{\sigma^2 \log^{2/\gamma_1}(\Gamma \vee \Delta)}{n \kappa^{d+2}}\right\}^{1/2},
\end{equation}
then the kernel bandwidth is consequently $h \asymp \kappa$.  In this case, our detection delay is not only consistent but also optimal.  To be specific, due to Assumption \ref{assump-snr},
\[
    w/\Delta \asymp \Delta^{-1} \left[\frac{\sigma^2 \log(\Gamma \vee \Delta)}{n \kappa^{d+2}} \vee \left\{\frac{\sigma^2 \log^{2/\gamma_1}(\Gamma \vee \Delta)}{n \kappa^{d+2}}\right\}^{1/2}\right] \leq a_n^{-1} \to 0.
\]
Moreover, in the case when $\gamma_1 = 2$, i.e.~when the data are independent, it follows from \cite{padilla2021optimal} that the detection delay \eqref{eq-optimal-delay} is optimal.

\section{Numerical experiments}\label{sec-numerical}

Throughout the simulations, we consider four generative scenarios.  In each scenario, we generate 50 data sets and report the average delay of our proposed method.  In each scenario the data are generated as
\begin{align*}
    & Y_{t,i} =  Z_{t,i} Y_{t,i}(1) +  (1-Z_{t,i}) Y_{t,i}(0), \quad Y_{t,i}(1) = \mu_{0}(X_{t,i}) + \tau_t(X_{t,i}) +  \epsilon_{t,i}(1), \\
    & Y_{t,i}(0) = \mu_{0}(X_{t,i}) + \epsilon_{t,i}(0), \quad \mathbb{P}(Z_{t,i}=1| X_{t,i}) = \pi(X_{t,i}), \\
    & (\epsilon_{1,i}(l), \ldots, \epsilon_{T,i}(l))^{\top} \overset{\mathrm{ind}}{\sim} F, \quad l=0,1, \quad \mbox{and} \quad X_{t,i} \overset{\mathrm{ind}}{\sim}  \mathrm{Unif}([0,1]^d).
\end{align*}
We consider $d \in \{3,6\}$, $T = 100$, $n = 40$,  the change point location $\Delta = 50$ and four different scenarios described below.  For any $x \in [0, 1]^d$, we write $x = (x_1, \ldots, x_d)^{\top}$.

\medskip
\noindent \textbf{Scenario 1.} For $x \in [0, 1]^d$, consider a randomised experiment where $\pi(x) = 0.5$.  Let the mean function be $\mu_{0}(x) = \|x\|^2$ and the treatment effect functions be $\tau_t(x) = x_1\,1\{t > \Delta\}$.  The errors are $\epsilon_{t,i}(l) \overset{\mathrm{ind}}{\sim} \mathcal{N}(0,1)$, for $t \in \{1, \ldots, T\}$, $i \in \{1, \ldots, n\}$ and $l \in \{0, 1\}$.

\medskip
\noindent \textbf{Scenario 2.} For $x \in [0, 1]^d$, let the propensity score be $\pi(x) = 0.25x_1 (1-x_1)^{4}/  \mathrm{B}(2, 5)$, where $\mathrm{B}(2, 5)$ denotes the Beta(2, 5) function.  Let $\mu_{0}(x) = 2x_1-1$ and $\tau_t(x) = (x_1 + x_2/2)\,1\{t > \Delta\}$.   Let $\varepsilon_{s, i}(l) \overset{\mathrm{ind}}{\sim} \mathcal{N}(0, 1)$, $s \in \{1, \ldots, T\}$, $l \in \{0, 1\}$ and $i \in \{1, \ldots, n\}$.  Let the errors be $\epsilon_{s,i}(l) = \varepsilon_{s,i}(l)$, for $s \in \{1, 2, 3\}$ and $l \in \{0, 1\}$.  For $t > 3$, let $\epsilon_{t,i}(l) = \{\varepsilon_{t,i}(l) + \varepsilon_{t-1,i}(l) + \varepsilon_{t-2,i}(l) + \varepsilon_{t-3,i}(l)\}/4$.  This model appeared as Scenario 1 in \cite{wager2018estimation} for a single time point.

\medskip
\noindent \textbf{Scenario 3.} For $x \in [0, 1]^d$, consider a randomised experiment where  $\pi(x) = 0.5$.  Let the mean function be $\mu_{0}(x) = \cos(100/x_1)$ and 
\[
    \tau_t(x) = \begin{cases}
        \left[1 + \frac{1}{1 + \exp\{-20(x_1 - 1/3)\}}\right]	\left[1 + \frac{1}{1 + \exp\{-20(x_2- 1/3)\}}\right]	&  t \leq \Delta, \\
        2\left[1 + \frac{1}{1 + \exp\{-20(x_1 - 1/3)\}}\right]	\left[1 + \frac{1}{1 + \exp\{-20(x_2- 1/3)\}}\right] & t > \Delta.
    \end{cases}
\] 
The errors $\epsilon_{t,i}(l)$ are generated as in \textbf{Scenario 2}.  This model is inspired by Scenario 2 of \cite{wager2018estimation}.

\medskip
\noindent \textbf{Scenario 4.}  Let $\beta = (\beta_1, \ldots, \beta_d)^{\top} \in R^d$ with $\beta_j = 1\{j = 1\} - 1\{j = 2\} + 1\{j = 3\}$.  Define the propensity score as $\pi(x) = \Phi(x^{\top} \beta)$, where $\Phi(\cdot)$ is the cumulative distribution function of $\mathcal{N}(0, 1)$. We also set $\mu_{0}(x) = \cos(300/x_1)$, and $\tau_t(x) = (2x_1 +3x_2) \,1\{t > \Delta\}$, $x\in [0,1]^d$.  Let $\varepsilon_{s, i}(l) \overset{\mathrm{ind}}{\sim} \mathcal{N}(0, 1)$, $s \in \{1, \ldots, T\}$, $l \in \{0, 1\}$ and $i \in \{1, \ldots, n\}$.  Let the errors be $\epsilon_{s,i}(l) = \varepsilon_{s,i}(l)$, for $s \in \{1, 2, 3, 4\}$ and $l \in \{0, 1\}$.  For $t > 4$, let $\epsilon_{t,i}(l) = \{\varepsilon_{t,i}(l) + \varepsilon_{t-1,i}(l) + \varepsilon_{t-2,i}(l) + \varepsilon_{t-3,i}(l) + \varepsilon_{t-4,i}(l)\}/8$.  

\medskip
\noindent \textbf{Choice of tuning parameters.} There are three tuning parameters for implementing our method. The first tuning parameter is the window width $w$. We recommend choosing $w$ as large as tolerable for delay in detecting a change point.  For our experiments we let $w = 3$ ($w=4$ in Section \ref{sec:exp2} in the Appendix).  As for the thresholding parameter $\varepsilon$, we first specify a desired $\Gamma >0$  and then select $\varepsilon$ such that $\mathbb{E}_{\infty}(\widehat{\Delta}) \geq \Gamma$.  Considering $\Gamma \in \{20, 30\}$, we conduct the calibration  by generating data from the model before the change point and calculating the expectation of the first false alarm  based on 100 Monte Carlo simulations.  As for the bandwidth $h$, we also recommend setting it to some constant value since our theory suggests $h \asymp \kappa$. To assess the sensitivity to $h$, in our experiments  we let  $h \in \{4,20\}$.  With regards to the propensity score,  in Scenarios 1 and 3 we use $\hat{\pi}(x) = \sum_{i=1}^{n} \sum_{t=1}^{T}Z_{t,i} /(nT)$,   and in Scenarios  2 and 4 we use the output of Algorithm \ref{alg-prop}.

\medskip
\noindent \textbf{Competitors.}  We are not aware of any existing competing method. To this end, we modify Algorithm \ref{alg-main} to propose a difference-based kernel method (DK).  DK adopts the same strategy as Algorithm \ref{alg-main}, only replacing  $\widehat{\tau}_{t-2w, t-w}(\cdot) $ and  $\widehat{\tau}_{t-w, t}(\cdot)$ with 
\[
    \widehat{\tau}_{t-2w, t-w}(\cdot) = \frac{\sum_{l = t-2w+1}^{t-w} \sum_{i = 1}^n Y_{l, i} 1_{ \left\{Z_{l,i}=1\right\}} \mathpzc{k}\left(\frac{X_{l, i} - \cdot}{h}\right)}{\sum_{l = t-2w+1}^{t-w} \sum_{i = 1}^n 1_{ \left\{Z_{l,i}=1\right\}}  \mathpzc{k}\left(\frac{X_{l, i} - \cdot}{h}\right)}- \frac{\sum_{l = t-2w+1}^{t-w} \sum_{i = 1}^n Y_{l, i} 1_{ \left\{Z_{l,i}=0\right\}} \mathpzc{k}\left(\frac{X_{l, i} - \cdot}{h}\right)}{\sum_{l = t-2w+1}^{t-w} \sum_{i = 1}^n 1_{ \left\{Z_{l,i}=0\right\}}  \mathpzc{k}\left(\frac{X_{l, i} - \cdot}{h}\right)}
\]
and
\[
    \widehat{\tau}_{t-w, t}(\cdot) = \frac{\sum_{l = t-w+1}^t \sum_{i = 1}^n  Y_{l, i}  1_{ \left\{Z_{l,i}=1\right\}}   \mathpzc{k}\left(\frac{X_{l, i} - \cdot}{h}\right)}{\sum_{l =t-w+1}^t \sum_{i = 1}^n 1_{ \left\{Z_{l,i}=1\right\}} \mathpzc{k}\left(\frac{X_{l, i} - \cdot}{h}\right)}-\frac{\sum_{l = t-w+1}^t \sum_{i = 1}^n  Y_{l, i}  1_{ \left\{Z_{l,i}=0\right\}}   \mathpzc{k}\left(\frac{X_{l, i} - \cdot}{h}\right)}{\sum_{l =t-w+1}^t \sum_{i = 1}^n 1_{ \left\{Z_{l,i}=0\right\}} \mathpzc{k}\left(\frac{X_{l, i} - \cdot}{h}\right)},
\]
respectively.  The tuning parameters are chosen as detailed above.  Our proposed method in Algorithm \ref{alg-main} in fact constructs a nonparametric estimator of the CATE function, while the competitor DK constructs two nonparametric estimators of the two mean functions for control and treatment groups, separately, before taking their difference to obtain an estimator of the CATE function.  In Section \ref{sec-problem-formulation} we have emphasised that Algorithm \ref{alg-main} allows for arbitrary mean functions and depends on fewer tuning parameters.  With the comparisons of DK, we solidate this argument.  See Section \ref{sec:ad_fig} for more discussions.

\begin{table}[t!]
	\centering
	\caption{  	\label{tab1} Delay averaging over 50 Monte Carlo simulations for different scenarios  and choices of tuning parameters. }
	\medskip
	\setlength{\tabcolsep}{3.5pt}
	\begin{small}
		\begin{tabular}{rr|cc|cc||cc|cc} 
			\hline
			&           &               \multicolumn{4}{c}{Scenario 1}                &    \multicolumn{4}{c}{Scenario 2}\\ 
			$d$  & $h$      & Delay                                     & Delay                                      & Delay                                                  & Delay                  & Delay                                     & Delay                                      & Delay                                                  & Delay  \\ 
			&                  &          ($\Gamma =20$)     &     ($\Gamma =20$)          &          ($\Gamma =40$)               &($\Gamma =40$) &          ($\Gamma =20$)     &     ($\Gamma =20$)          &          ($\Gamma =40$)               &($\Gamma =40$) \\
			&              &                 Alg 1                 &              DK                                   &        Alg 1                                             &  DK                       &                 Alg 1                          &       DK                                  &                         Alg 1                 &  DK\\
			\hline    
			3    &    20         &8.4 (8.7)                      &             11.7 (14.3)                    & 17.7 (18.6)                                & 38.9 (30.6)                         &               19.8 (20.2)         &          33.3 (30.7)        &              28.5 (23.9)          &34.3 (29.9)\\
			3     &    4          &    21.3 (22.1)          &             22.3 (25.8)                      &     33.7 (26.3)                                     &          33.7 (32.6)           &               14.2 (20.8)       &       17.0 (27.5)               &        26.9 (28.3)  &27.5 (31.1)\\
			6     &    20         &      12.4 (10.1)            &         10.7 (9.6)                          &      16.6 (19.7)                              &   26.2 (30.1)                 &            13.8 (24.1)        &             14.5 (24.6)                     &             26.3(28.5)         &29.8 (31.2)\\
			6      &    4          &    11.3 (8.9)               &        16.3 (23.3)                       &      21.4 (8.9)                          &  29.9 (29.8)                        &              16.8 (26.7)              &        17.2 (26.7)            &       30.1 (30.0)                 & 27.3 (31.8)\\			          			          			          
			\hline
			\hline
			&           &               \multicolumn{4}{c}{Scenario 3}                &    \multicolumn{4}{c}{Scenario 4}\\ 
			$d$  & $h$      & Delay                                     & Delay                                      & Delay                                                  & Delay                  & Delay                                     & Delay                                      & Delay                                                  & Delay  \\ 
			&                  &          ($\Gamma =20$)     &     ($\Gamma =20$)          &          ($\Gamma =40$)               &($\Gamma =40$) &          ($\Gamma =20$)     &     ($\Gamma =20$)          &          ($\Gamma =40$)               &($\Gamma =40$) \\
			&                             &                 Alg 1                 &              DK                                   &        Alg 1                                             &  DK                       &                 Alg 1                          &       DK                                  &                         Alg 1                 &  DK\\
			\hline
			3     &    20         &  15.2 (27.6)                 &     15.5 (27.8)                          &25.3 (32.8)                                        &     25.4 (32.8) &      15.3 (27.5)                   &               18.8 (29.7)                   &            22.9 (30.2)              &25.5 (32.8)\\
			3     &    4          &       11.0 (21.4)                   &    22.1 (31.5)                        &      23.6 (21.4)                              &       38.9 (34.2)     &           22.0 (32.6)              &      23.4 (28.8)                  &            25.4 (32.8)             & 25.4 (32.8)\\
			6     &    20         &       18.9 (28.5)                  &             22.1 (31.5 )         &       25.5 (34.3)                                & 28.8 (33.7)   &                        18.7 (29.8)           &              22.1 (31.5)                &            28.8 (33.7)                       &32.2 (34.2)\\
			6     &    4          &        8.7 (20.6)                    &          15.4 (27.5 )               &    25.5 (32.8)                            &   25.5 (32.8)             &                 15.3 (27.5)                  &            15.4 (27.5)                  &          27.2 (33.2)                  & 28.8 (34.7)\\				   	     
			\hline
			\hline			
		\end{tabular}
	\end{small}
\end{table}

The results in Table \ref{tab1} show that the proposed method perform reasonably well  across all the scenarios considered. In particular, Algorithm \ref{alg-main} clearly  outperforms DK in Scenarios 3 and 4, where the functions $\mu_t(\cdot)$ are not well behaved, making the estimation of  the functions $E(Y_{t,i}(1) |X_{t,i}=x)$ and   $E(Y_{t,i}(0) |X_{t,i}=x)$ challenging. 
We also see in Table \ref{tab1} that when the functions $E(Y_{t,i}(1) |X_{t,i}=x)$ and   $E(Y_{t,i}(0) |X_{t,i}=x)$ are smooth as in Scenarios 1 and 2, the competing method DK and Algorithm \ref{alg-main} perform more similarly, which is  expected since in those cases the DK method is indeed a reasonable choice.

\section{Conclusion}

In this paper, we are concerned with detecting abrupt changes in CATE functions, with temporal dependence.  A sliding window technique, coupled with Nadaraya--Watson estimators, leads us to a change point estimator.  We have shown that such procedure is theoretically consistent and computationally efficient.  

For an online change point detection problem, controlling false alarms is as important as minimising the detection delay.  We considered lower bounding the ARL by $\Gamma$ in this work.  Another popular way is to upper bound the overall Type-I error.  The method and theory developed in this paper can be straightforwardly extended to the overall Type-I error control, provided that the data are assumed to be independent across time.  Considering change point detection in dynamic treatment effects under this more conservative control remains an interesting but challenging problem.

Other research directions in dynamic treatment effects change point detection include considering more general classes of functions for the CATE, such as the class of bounded variation functions such as in \cite{mammen1997locally}.  Another natural extension is to allow for high-dimensional covariates, perhaps employing some type of  $\ell_1$ or $\ell_2$ regularisation.


\newpage
\bibliographystyle{agsm}
\bibliography{ref}

\appendix

\section{Proofs}

\subsection{Additional notation}\label{sec-additional-notation}

Recall that for any $s \in N^*$, $i \in \{1, \ldots, n\}$ and $x \in \mathcal{X}$,
\[
    E\left[Y_{s, i} \left\{\frac{Z_{s, i}}{\pi(x)} - \frac{1 - Z_{s, i}}{1 - \pi(x)}\right\} \Big| X_{s, i} = x\right]= \tau_s(x).
\]
For any $s \in N^*$ and $i \in \{1, \ldots, n\}$, define
\[
    \widetilde{Y}_{s, i} = Y_{s, i}\left\{ \frac{Z_{s, i}}{\pi(X_{s, i})} - \frac{1 - Z_{s, i}}{1 - \pi(X_{s, i})} \right\}  \quad \mbox{and} \quad \widehat{Y}_{s, i} = Y_{s, i}\left\{\frac{Z_{s, i}}{\widehat{\pi}(X_{s, i})} - \frac{1 - Z_{s, i}}{1 - \widehat{\pi}(X_{s, i})} \right\}.
\]

For any integer $t > w$ and $x \in \mathcal{X}$, define the kernel estimators of $\widetilde{Y}_{s, i}$ and $\widehat{Y}_{s, i}$ as
\[
    \widetilde{\Psi}_{t-w, t}(x) = \frac{1}{nw h^p} \sum_{s = t-w+1}^t \sum_{i = 1}^n \widetilde{Y}_{s, i} \mathpzc{k}\left(\frac{X_{s, i} - x}{h}\right)
\]
and 
\[
    \widehat{\Psi}_{t-w, t}(x) = \frac{1}{nw h^p} \sum_{s = t-w+1}^t \sum_{i = 1}^n \widehat{Y}_{s, i} \mathpzc{k}\left(\frac{X_{s, i} - x}{h}\right),
\]
respectively.

For any integer $t > w$ and any $x \in \mathcal{X}$, define 
\[
    \widehat{g}_{t-w, t}(x) = \frac{1}{nwh^d} \sum_{s = t-w+1}^t \sum_{i = 1}^n \mathpzc{k}\left(\frac{X_{s, i} - x}{h}\right) \quad \mbox{and} \quad \tau_{t-w, t}(x) = \frac{1}{w}\sum_{s = t-w+1}^t \tau_s(x).
\]

We defined the event
\[
    \mathcal{G} = \{\|\widehat{\pi} - \pi\|_{\infty} < \xi_w < c_{\mathrm{prop}}/2\}.
\]

\subsection[]{Proof of Theorem \ref{thm-main}}

\begin{proof}[Proof of Theorem \ref{thm-main}]

This proof is conducted in a few steps. 

\medskip
\noindent \textbf{Step 1.}  For $\gamma \in (0, 1/2)$, let 
\begin{equation}\label{eq-Q-cond}
    Q = C\max\{\Gamma/(1- \gamma), \, 2w\},
\end{equation}
where $C > 1$ is an absolute constant.  In this step, we are to show that, for any integer $u \in N$, the event
    \[
        \mathcal{E}_u = \left\{\|\widehat{\tau}_{t-w, t} - \tau_{t-w, t}\|_{\infty} \leq \xi_w/2, \, \forall u+w < t \leq u+Q \mbox{ and } \Delta \notin [u, u+Q]\right\}
    \]
    holds with probability at least $1 - \gamma/2$, where $\xi_w$ is defined in Proposition \ref{prop-mathcal-E-event}.
    
Let 
    \[
        \widehat{\mathcal{X}} = \{X_{r, i}, \, i = 1, \ldots, n, \, r = 1, \ldots, 2w\}.
    \]
    Proposition \ref{prop-mathcal-F} shows that 
    \[
        \mathcal{F} = \bigg\{\sup_{x \in \mathcal{X}} \min_{\hat{x} \in \widehat{\mathcal{X}}} |\{\tau_{\Delta}(x) - \tau_{\Delta+1}(x)\} - \{\tau_{\Delta}(\hat{x}) - \tau_{\Delta+1}(\hat{x})\}| \leq \xi_w/2\bigg\}
    \]
    holds with probability at least $1 - \gamma/2$.  
    
\medskip
\noindent \textbf{Step 2: Average run length control.}  Under Assumption \ref{assump-change-point-scenarios}\textbf{(a.)}, in the event $\mathcal{E}_0 \cap \mathcal{F}$, it holds that 
\begin{align*}
    & \max_{t = 2w+1}^Q \max_{x \in \widehat{\mathcal{X}}}|\widehat{\tau}_{t-2w, t-w}(x) - \widehat{\tau}_{t-w, t}(x)| \\
    \leq & \max_{t = 2w+1}^Q \max_{x \in \widehat{\mathcal{X}}} |\widehat{\tau}_{t-2w, t-w}(x) - \tau_{t-2w, t-w}(x)| + \max_{t = 2w+1}^Q \max_{x \in \widehat{\mathcal{X}}} |\widehat{\tau}_{t-w, t}(x) - \tau_{t-w, t}(x)| \\
    \leq & \max_{t = 2w+1}^Q \|\widehat{\tau}_{t-2w, t-w} - \tau_{t-2w, t-w}\|_{\infty} + \max_{t = 2w+1}^Q \|\widehat{\tau}_{t-w, t} - \tau_{t-w, t}\|_{\infty} \leq \xi_w < \varepsilon,
\end{align*}
where the first inequality is due to the triangular inequality, the third is due to the definition of $\mathcal{E}_0$ and the last is due to the definition of $\varepsilon$.  We therefore have that, for any $t \in (0, Q]$, $\widehat{\Delta} > t$.  It therefore follows that
\begin{align*}
	E_{\infty} (\widehat{\Delta}) & = \int_0^{\infty} \mathbb{P}(\widehat{\Delta} > t) \dint t \geq \int_0^Q \,\mathbb{P}(\widehat{\Delta} > t) \dint t \\
	& \geq Q \inf_{t \in [0, Q]} \mathbb{P}\{\widehat{\Delta} > t\} \geq Q \,\mathbb{P}\{\widehat{\Delta} > Q\} \geq Q \,\mathbb{P}\{\mathcal{E}_0 \cap \mathcal{F}\} \geq Q(1 - \gamma) \geq \Gamma,
\end{align*}
where the last inequality follows from  \eqref{eq-Q-cond}.
    
\medskip    
\noindent \textbf{Step 3: Detection delay control.} In order to control the detection delay, we consider a specific choice of $\gamma = \Delta^{-c}$.  Under Assumption \ref{assump-change-point-scenarios}\textbf{(a.)}, since $\mathcal{X}$ is a compact set, we let $x^* = \sup_{x \in \mathcal{X}} |\tau_{\Delta}(x) - \tau_{\Delta+1}(x)|$.  Consider the event $\mathcal{E}_{\Delta-w} \cap \mathcal{F}$, it holds that
\begin{align}
    & \max_{x \in \widehat{\mathcal{X}}} |\widehat{\tau}_{\Delta - w, \Delta}(x) - \widehat{\tau}_{\Delta, \Delta + w}(x)| \nonumber \\
    \geq & \max_{x \in \widehat{\mathcal{X}}} |\tau_{\Delta}(x) - \tau_{\Delta + 1}(x)| - \max_{x \in \widehat{\mathcal{X}}} |\widehat{\tau}_{\Delta - w, \Delta}(x) - \tau_{\Delta}(x)| - \max_{x \in \widehat{\mathcal{X}}} |\widehat{\tau}_{\Delta, \Delta + w}(x) - \tau_{\Delta + 1}(x)| \nonumber \\
    \geq & \max_{x \in \widehat{\mathcal{X}}} |\tau_{\Delta}(x^*) - \tau_{\Delta}(x) + \tau_{\Delta + 1}(x) - \tau_{\Delta+1}(x^*) + \tau_{\Delta+1}(x^*) - \tau_{\Delta}(x^*)| \nonumber \\
    & \hspace{2cm} -  \|\widehat{\tau}_{\Delta - w, \Delta} - \tau_{\Delta}\|_{\infty} - \|\widehat{\tau}_{\Delta, \Delta + w} - \tau_{\Delta+1}\|_{\infty} \nonumber \\
    \geq & \|\tau_{\Delta} - \tau_{\Delta+1}\|_{\infty} - \min_{x \in \widehat{\mathcal{X}}} |\tau_{\Delta}(x^*) - \tau_{\Delta}(x) + \tau_{\Delta+1}(x) - \tau_{\Delta+1}(x^*)| - \xi_w \nonumber \\
    \geq & \kappa - 2\xi_w \geq \varepsilon, \label{eq-proof-thm-main-1}
\end{align}
where the first inequality is due to the triangular inequality, the third is due to the definition of the event $\mathcal{E}_{\Delta - w}$, the fourth is due to the definitions of $x^*$ and $\mathcal{F}$, and the last is due to Assumption \ref{assump-snr} and the condition of $\varepsilon$.  We then have that $\widehat{\Delta} \leq \Delta + w$.

\end{proof}

\begin{proposition}\label{prop-mathcal-F}
Under Assumption \ref{assume-model-1}, we have that
\[
    \mathbb{P}\{\mathcal{F}\} \geq 1 - \gamma/2,
\]
provided that $n \geq \log(\gamma/2)/\log(1-c_d\xi_w)$, where $\xi_w$ is defined in Proposition \ref{prop-mathcal-E-event}.
\end{proposition}

\begin{proof}[Proof of Proposition ref{prop-mathcal-F}]
We consider the event
\[
    \mathcal{F}^c = \bigg\{\sup_{x \in \mathcal{X}} \min_{\hat{x} \in \widehat{\mathcal{X}}} |\{\tau_{\Delta}(x) - \tau_{\Delta+1}(x)\} - \{\tau_{\Delta}(\hat{x}) - \tau_{\Delta+1}(\hat{x})\}| > \xi_w/2\bigg\}.
\]
If $\mathcal{F}^c$ holds, then there exists $x_0 \in \mathcal{X}$ such that for any $\hat{x} \in \widehat{\mathcal{X}}$, it holds that
\[
    |\{\tau_{\Delta}(x_0) - \tau_{\Delta+1}(x_0)\} - \{\tau_{\Delta}(\hat{x}) - \tau_{\Delta+1}(\hat{x})\}| \geq \xi_w/2.
\]
Due to Assumption \ref{assume-model-1}\textbf{(c.)}, $\tau_{\Delta}(\cdot)$ and $\tau_{\Delta+1}(\cdot)$ are both Lipschitz functions, implying that
\[
    |x_0 - \hat{x}| \geq \xi_w/(4C_{\mathrm{Lip}}).  
\]
This means that 
\[
    \mathcal{A} = \widehat{\mathcal{X}} \cap \{s \in \mathcal{X}: \, \|s - x_0\| \leq \xi_w/(4C_{\mathrm{Lip}})\} = \emptyset.    
\]
Due to Assumption \ref{assume-model-1}\textbf{(b.)}, $g(\cdot)$ is lower bounded by $c_{g, 1}$ in $\mathcal{X}$.  It follows from the independence across subjects that
\[
    \mathbb{P}\{\mathcal{F}^c\} \leq \mathbb{P}\{\mathcal{A}\} \leq (1 - c_d \xi_w^d)^n \leq \gamma/2,
\]
where $0 < c_d < 1$ is an absolute constant but depending on the fixed dimensionality $d$.  This includes the situation when $x_0$ happens to be at the corner of $\mathcal{X} = [0, 1]^d$.

\end{proof}

\begin{proposition}\label{prop-mathcal-E-event}
Under Assumptions~\ref{assume-model-1}, \ref{assump-causal} and \ref{assump-kernel}, with 
\[
    \xi_w \asymp \left\{\frac{\sigma^2 \log^{2/\gamma_1}(Qw/\gamma)}{nw^2}\right\}^{1/(d+2)} \vee \left\{\frac{\sigma^2 \log(Qw/\gamma)}{nw}\right\}^{1/(d+2)} \vee \|\widehat{\pi} - \pi\|_{\infty}
\]
and
\[
    h \asymp \left\{\frac{\sigma^2 \log^{2/\gamma_1}(Qw/\gamma)}{nw^2}\right\}^{1/(d+2)} \vee \left\{\frac{\sigma^2 \log(Qw/\gamma)}{nw}\right\}^{1/(d+2)},
\]
for any $u \in  N$, it holds that 
\[
    \mathbb{P}\{\mathcal{E}_u\} > 1 - \gamma/2.
\]
\end{proposition}

\begin{proof}[Proof of Proposition \ref{prop-mathcal-E-event}]
Based on the notation introduced in Section  \ref{sec-additional-notation}, for any integer $t > w$ and any $x \in \mathcal{X}$, we can rewrite $\widehat{\tau}_{t-w, t}$ as 
\[
    \widehat{\tau}_{t-w, t}(x) = \frac{\widehat{\Psi}_{t-w, t}(x)}{\widehat{g}_{t-w, t}(x)}.
\]
We then have that 
\begin{align*}
    & |\widehat{\tau}_{t-w, t}(x) - \tau_{t-w, t}(x)| \\
    \leq & \left|\frac{\widehat{\Psi}_{t-w, t}(x) - \widetilde{\Psi}_{t-w, t}(x)}{\widehat{g}_{t-w, t}(x)}\right| +  \left|\frac{\widetilde{\Psi}_{t-w, t}(x) - E\left\{\widetilde{\Psi}_{t-w, t}(x)\right\}}{\widehat{g}_{t-w, t}(x)}\right| + \left|\frac{ E\left\{\widetilde{\Psi}_{t-w, t}(x)\right\}}{\widehat{g}_{t-w, t}(x)} - \tau_{t-w, t}(x)\right|\\
    \leq & \left|\frac{\widehat{\Psi}_{t-w, t}(x) - \widetilde{\Psi}_{t-w, t}(x)}{g(x) - \|\widehat{g}_{t-w, t} - g\|_{\infty}}\right| + \left|\frac{\widetilde{\Psi}_{t-w, t}(x) -  E\left\{\widetilde{\Psi}_{t-w, t}(x)\right\}}{g(x) - \|\widehat{g}_{t-w, t} - g\|_{\infty}}\right| + \left|\frac{E\left\{\widetilde{\Psi}_{t-w, t}(x)\right\}}{\widehat{g}_{t-w, t}(x)} - \tau_{t-w, t}(x)\right| \\
    = & (I) + (II) + (III),
\end{align*}
where term $(I)$ is dealt with in Lemmas~\ref{lem1}, \ref{lem-abs-epsilon-bound} and \ref{lem-ghat-g-diff}, term $(II)$ is dealt with in Lemmas~\ref{lem-psi-Epsi-bound} and \ref{lem-ghat-g-diff}, and term $(III)$ is dealt with in Lemma \ref{lem-epis-tau-diff-bound}.

Taking 
\[
    h \asymp \left\{\frac{\sigma^2 \log^{2/\gamma_1}(Qw/\gamma)}{nw^2}\right\}^{1/(d+2)} \vee \left\{\frac{\sigma^2 \log(Qw/\gamma)}{nw}\right\}^{1/(d+2)},
\]
we conclude the proof.

\end{proof}

\section{Technical details}
\begin{lemma}\label{lem-justifying-estimator}
Under Assumptions~\ref{assume-model-1}, \ref{assump-causal}\textbf{(b.)} and \textbf{(c.)}, it holds that 
    \[
        E\left[Y_{t, 1} \left\{\frac{Z_{t, 1}}{\pi(x)} - \frac{1 - Z_{t, 1}}{1 - \pi(x)}\right\} \Big| X_{t, 1} = x\right] = \tau_t(x), \quad t \in N^*.
    \]
\end{lemma}

\begin{proof}[Proof of Lemma \ref{lem-justifying-estimator}]
For any $t \in \mathbb{N}^*$, it holds that
    \begin{align*}
        & E\left\{Y_{t, 1} \left(\frac{Z_{t, 1}}{\pi(x)} - \frac{1 - Z_{t, 1}}{1 - \pi(x)}\right) \Big| X_{t, 1} = x\right\} =  E\left\{Y_{t, 1}(1)\frac{Z_{t, 1}}{\pi(x)} - Y_{t, 1}(0)\frac{1 - Z_{t, 1}}{1 - \pi(x)}\Big| X_{t, 1} = x\right\} \\
        = & E\{Y_{t, 1}(1) | X_{t, 1} = x\} \frac{\mathbb{P}\{Z_{t, 1} = 1 | X_{t, 1} = x\}}{\pi(x)} - E\{Y_{t, 1}(0) | X_{t, 1} = x\} \frac{\mathbb{P}\{Z_{t, 1} = 0 | X_{t, 1} = x\}}{1 - \pi(x)} \\
        = & E\{Y_{t, 1}(1) - Y_{t, 1}(0)| X_{t, 1} = x\} = \tau_t(x),
    \end{align*}
    where the first identity follows from the fact that $Z_{t, 1} \in \{0, 1\}$, the second follows from the conditional independence detailed in Assumption \ref{assump-causal}\textbf{(b.)}, the third follows from Assumption \ref{assump-causal}\textbf{(c.)} and the last one follows from Assumption \ref{assume-model-1}\textbf{(e.)}.
\end{proof}

\begin{lemma} \label{lem1}
For any integer $t > w$ and any $x \in \mathcal{X}$,  we have that 
    \begin{align*}
        & \left|\widehat{\Psi}_{t-w, t}(x) - \widetilde{\Psi}_{t-w, t}(x)\right| \\
        \leq &\frac{2c_{\mu}\|\widehat{\pi} - \pi\|_{\infty}}{c_{\mathrm{prop}}^2} + \frac{2\|\widehat{\pi} - \pi\|_{\infty}}{c_{\mathrm{prop}}^2} \frac{1}{nw h^d}\sum_{s = t-w+1}^t \sum_{i = 1}^n |\epsilon_{s, i}(Z_{s, i})| \mathpzc{k}\left(\frac{X_{s, i} - x}{h}\right).
    \end{align*}
\end{lemma}

\begin{proof}
Due to the definitions in Section \ref{sec-additional-notation}, we have that for any $s \in N^*$ and $i \in \{1, \ldots, n\}$, 
\begin{align*}
    \widehat{Y}_{s, i} - \widetilde{Y}_{s, i} & = \frac{Y_{s, i}\{Z_{s, i} - \widehat{\pi}(X_{s, i})\}}{\widehat{\pi}(X_{s, i})\{1 - \widehat{\pi}(X_{s, i})\}} - \frac{Y_{s, i}\{Z_{s, i} - \pi(X_{s, i})\}}{\widehat{\pi}(X_{s, i})\{1 - \widehat{\pi}(X_{s, i})\}} \\
    & \hspace{1cm} + \frac{Y_{s, i}\{Z_{s, i} - \pi(X_{s, i})\}}{\widehat{\pi}(X_{s, i})\{1 - \widehat{\pi}(X_{s, i})\}} - \frac{Y_{s, i}\{Z_{s, i} - \pi(X_{s, i})\}}{\pi(X_{s, i})\{1 - \pi(X_{s, i})\}} \\
    & = \frac{Y_{s, i}\{\pi(X_{s, i}) - \widehat{\pi}(X_{s, i})\}}{\widehat{\pi}(X_{s, i})\{1 - \widehat{\pi}(X_{s, i})\}}\left[1 + \frac{\{Z_{s, i} - \pi(X_{s, i})\}\{1 - \pi(X_{s, i}) - \widehat{\pi}(X_{s, i})\}}{\pi(X_{s, i})\{1 - \pi(X_{s, i})\}}\right] \\
    & = Y_{s, i}\{\pi(X_{s, i}) - \widehat{\pi}(X_{s, i})\} \left[\frac{Z_{s, i}}{\pi(X_{s, i}) \widehat{\pi}(X_{s, i})} + \frac{1 - Z_{s, i}}{\{1 - \pi(X_{s, i})\}\{1 - \widehat{\pi}(X_{s, i})\}}\right].
\end{align*}
We then have
\begin{align*}
    & \left|\widehat{\Psi}_{t-w, t}(x) - \widetilde{\Psi}_{t-w, t}(x)\right| = \frac{1}{nwh^d}\left|\sum_{s = t-w+1}^t \sum_{i = 1}^n (\widehat{Y}_{s, i} - \widetilde{Y}_{s, i}) \mathpzc{k}\left(\frac{X_{s, i} - x}{h}\right) \right| \\
    \leq & \frac{2\|\widehat{\pi} - \pi\|_{\infty}}{c_{\mathrm{prop}}^2} \frac{1}{nwh^d}\sum_{s = t-w+1}^t \sum_{i = 1}^n |\mu_{s, Z_{s, i}}(X_{s, i})| \mathpzc{k}\left(\frac{X_{s, i} - x}{h}\right) \\
    & \hspace{1cm} + \frac{2\|\widehat{\pi} - \pi\|_{\infty}}{c_{\mathrm{prop}}^2} \frac{1}{nwh^d}\sum_{s = t-w+1}^t \sum_{i = 1}^n |\epsilon_{s, i}(Z_{s, i})| \mathpzc{k}\left(\frac{X_{s, i} - x}{h}\right) \\
    \leq & \frac{2c_{\mu}\|\widehat{\pi} - \pi\|_{\infty}}{c_{\mathrm{prop}}^2} + \frac{2\|\widehat{\pi} - \pi\|_{\infty}}{c_{\mathrm{prop}}^2} \frac{1}{nwh^d}\sum_{s = t-w+1}^t \sum_{i = 1}^n |\epsilon_{s, i}(Z_{s, i})| \mathpzc{k}\left(\frac{X_{s, i} - x}{h}\right),
\end{align*}
where the first inequality holds due to the definition of event $\mathcal{G}$ and the second is due to Assumption \ref{assume-model-1}\textbf{(c.)}.
\end{proof}

\begin{lemma}\label{lem-abs-epsilon-bound}
For any $u \in N$, it holds that 
    \[
        \mathbb{P}\left\{\max_{t = u+w}^{u+Q} \sup_{x \in [0, 1]^d} \frac{1}{nwh^d}\sum_{s = t-w+1}^t \sum_{i = 1}^n |\epsilon_{s, i}(Z_{s, i})| \mathpzc{k}\left(\frac{X_{s, i} - x}{h}\right) > y\right\} \leq c \gamma,
    \]
    where 
    \[
      y \asymp \frac{\sigma \log^{1/\gamma_1}(Qw h^{-d}\gamma^{-1})}{(n)^{1/2}wh^{d/2}} \vee \frac{\sigma \log^{1/2}(Qw h^{-d}\gamma^{-1})}{(nw)^{1/2} h^{d/2}}, \quad \gamma_1 = \frac{2 + \gamma_{\alpha}}{2\gamma_{\alpha}} 
    \]
    and $C, c > 0$ are absolute constants.
\end{lemma}

\begin{proof}[Proof of Lemma \ref{lem-abs-epsilon-bound}]

Since $\mathcal{X}$ is compact, without loss of generality, we let $\mathcal{X} = [0, 1]^p$ in this proof for simplicity.  For $0 < r \leq L \wedge 1$, we can find $\{x_j\}_{j = 1}^M \subset [0, 1]^p$ with $M \leq r^{-p}h^{-p}$ such that for any $x \in [0, 1]^p$, there exists $j \in \{1, \ldots, M\}$, $\|x - x_j\|_{\infty} \leq rh$.  It follows from Assumption \ref{assump-kernel} that there exists a kernel function $\mathpzc{k}^*(\cdot)$, satisfying that for any $x, y \in [0, 1]^d$, $\|x - y\|_{\infty} \leq \delta \leq L$,
    \begin{equation}\label{eq-existence-kstar}
        |\mathpzc{k}(x) - \mathpzc{k}(y)| \leq \delta \mathpzc{k}^*(x),
    \end{equation}
    where $\mathpzc{k}^*(\cdot)$ satisfies Assumption \ref{assump-kernel}.  The existence of the kernel $\mathpzc{k}^*(\cdot)$ is shown in the proof of Theorem in \cite{hansen2008uniform} and we include it here for completeness.  If Assumption \ref{assump-kernel}\textbf{(a.)} holds, then $\mathpzc{k}^*(x) = C_{\mathrm{Lip}} 1\{\|x\|_{\infty} \leq 2L\}$, $x \in [0, 1]^p$, satisfies \eqref{eq-existence-kstar}.  If Assumption \ref{assump-kernel}\textbf{(b.)} holds, then $\mathpzc{k}^*(x) = C_{\mathrm{Lip}} 1\{\|x\|_{\infty} \leq 2L\} + \|x - L\|_{\infty}^{-v}1\{\|x\|_{\infty} > 2L\}$, $x \in [0, 1]^p$, satisfies \eqref{eq-existence-kstar}.

\medskip
\noindent \textbf{Step 1: Decomposition.}  For $j \in \{1, \ldots, M\}$, let $A_j = \{x: \|x - x_j\|_{\infty} \leq rh\}$.  Due to \eqref{eq-existence-kstar}, it holds that, for any $x \in A_j$, $s \in  N^*$ and $i \in \{1, \ldots, n\}$,
\[
    \left|\mathpzc{k}\left(\frac{X_{s, i} - x}{h}\right) - \mathpzc{k}\left(\frac{X_{s, i} - x_j}{h}\right)\right| \leq r \mathpzc{k}^*\left(\frac{X_{s, i} - x_j}{h}\right).
\]
We first fix an integer $t > w$.  Let
\[
    f(x) = f_t(x) = \frac{1}{nwh^d}\sum_{s = t-w+1}^e \sum_{i = 1}^n |\epsilon_{s, i}(Z_{s, i})| \mathpzc{k}\left(\frac{X_{s, i} - x}{h}\right)
\]
and
\[
    \widehat{f}(x) = \widehat{f}_t(x) = \frac{1}{nwh^d}\sum_{s = t-w+1}^t \sum_{i = 1}^n |\epsilon_{s, i}(Z_{s, i})| \mathpzc{k}^*\left(\frac{X_{s, i} - x}{h}\right).
\]
Note that for any $x \in [0, 1]^d$, 
\begin{align*}
    & E\left\{|\widehat{f}(x)|\right\} \leq \max_{z \in \{0, 1\}} E\left\{\frac{1}{h^d} |\epsilon_{1, 1,}(z)| \mathpzc{k}^* \left(\frac{X_{1, 1} - x}{h}\right)\right\} \\
    \leq & \max_{z \in \{0, 1\}} E\left[E\left\{|\epsilon_{1, 1}(z)| \big| X_{1, 1}\right\} \frac{1}{h^d} \mathpzc{k}^* \left(\frac{X_{1, 1} - x}{h}\right)\right] \leq \frac{\sigma C_{\mathpzc{k}}}{c_{g, 1}}. 
\end{align*}
which follows from Assumption \ref{assump-kernel} and Assumption \ref{assume-model-1}\textbf{(d.)}.  Identical arguments lead to 
\[
    E\left\{|f(x)|\right\} \leq \frac{\sigma C_{\mathpzc{k}}}{c_{g, 1}}, \quad x \in [0, 1]^d.
\]

We then have that 
\begin{align*}
    & \sup_{x \in A_j} |f(x) - E\{f(x)\}| = \sup_{x \in A_j} |f(x) - f(x_j) + f(x_j) -  E\{f(x_j)\} +  E\{f(x_j)\} - E\{f(x)\}|\\
    \leq & |f(x_j) -  E\{f(x_j)\}| + r \widehat{f}(x_j) + r E\{\widehat{f}(x_j)\} \\
    \leq & |f(x_j) -   E\{f(x_j)\}| + r |\widehat{f}(x_j) - E\{\widehat{f}(x_j)\}| + 2r E\{\widehat{f}(x_j)\} \\
    \leq & |f(x_j) -  E\{f(x_j)\}| + |\widehat{f}(x_j) - E\{\widehat{f}(x_j)\}| + 2r\frac{\sigma C_{\mathpzc{k}}}{c_{g, 1}},
\end{align*}
where the last inequality follows from $r \leq 1$.

We thus have that, for any $y > 2r\sigma C_{\mathpzc{k}}/c_{g, 1}$,
\begin{align*}
    & \mathbb{P}\left\{\sup_{x \in [0, 1]^d} |f(x) -  E\{f(x)\}| > y\right\} \leq M \max_{j = 1}^M \mathbb{P}\left\{\sup_{x \in A_j} |f(x) - E\{f(x)\}| > y\right\} \\
    \leq & M \max_{j = 1}^M \mathbb{P}\left\{|f(x_j) - E\{f(x_j)\}| > y - 2r\sigma C_{\mathpzc{k}}/c_{g, 1} \right\} \\
    & \hspace{1cm} + M \max_{j = 1}^M \mathbb{P}\left\{|\widehat{f}(x_j) - E\{\widehat{f}(x_j)\}| > y - 2r\sigma C_{\mathpzc{k}}/c_{g, 1}\right\} \\
    = & (I) + (II).
\end{align*}

\medskip
\noindent \textbf{Step 3: Bounding terms $(I)$ and $(II)$.}  Since both $\mathpzc{k}(\cdot)$ and $\mathpzc{k}^*(\cdot)$ satisfy Assumption \ref{assump-kernel}, we are only bounding term $(I)$ here.  Term $(II)$ can be dealt with using identical arguments.  

For any $x \in [0, 1]^p$, let
\[
    W_s = \frac{1}{nh^d}\sum_{i=1}^n|\epsilon_{s, i}(Z_{s, i})| \mathpzc{k}\left(\frac{X_{s, i} - x}{h}\right) -  E\left\{|\epsilon_{s, 1}(Z_{s, 1})| \mathpzc{k}\left(\frac{X_{s, 1} - x}{h}\right)\right\}.  
\] 
It follows from Assumption \ref{assume-model-1} that for any $\zeta > 0$
\[
    \mathbb{P}\{(nh^d)^{1/2}/\sigma |W_s| > \zeta\} \lesssim \exp(-C\zeta^2).
\]
It then follows from Theorem~1 in \cite{merlevede2011bernstein} that for $\eta \gtrsim \log(n)/n$, 
\begin{align*}
    & \mathbb{P}\left\{|f(x_j) -  E\{f(x_j)\}| > y - 2r\sigma C_{\mathpzc{k}}/c_{g, 1}\right\} = \mathbb{P}\left[\frac{\sigma}{(nh^d)^{1/2}     w}\left|\frac{n^{1/2}   }{\sigma}\sum_{s = t-w+1}^t W_s\right| > y - 2r\sigma C_{\mathpzc{k}}/c_{g, 1} \right] \\
    = & \mathbb{P}\left\{\left|\frac{n^{1/2}   }{\sigma}\sum_{s = t-w+1}^t W_s\right| > (y - 2r\sigma C_{\mathpzc{k}}/c_{g, 1}) n^{1/2}   w h^{d/2}/\sigma \right\} \\
    \leq & (w+1) \exp\left[-C\{(y - 2r\sigma C_{\mathpzc{k}}/c_{g, 1})n^{1/2}   w h^{d/2}/\sigma\}^{\gamma_1}\right] \\
    & \hspace{1cm} + \exp\left[-C \{(y - 2r\sigma C_{\mathpzc{k}}/c_{g, 1}) n^{1/2} w h^{d/2}/\sigma\}^2/w\right],
\end{align*}
where $C > 0$ is an absolute constant and $\gamma_1 = 2\gamma_{\alpha}/(2+\gamma_{\alpha})$.  

Since $M \leq r^{-p}h^{-p}$ by construction, we therefore have that
\begin{align*}
    (I) & \leq r^{-d}h^{-d}(w+1) \exp\left[-C\{(y - 2r\sigma C_{\mathpzc{k}}/c_{g, 1})   n^{1/2}  w h^{d/2}/\sigma\}^{\gamma_1}\right] \\
    & \hspace{1cm} + r^{-d}h^{-d}\exp\left[C \{(y - 2r\sigma C_{\mathpzc{k}}/c_{g, 1})n^{1/2} w h^{d/2}/\sigma\}^2/w\right]
\end{align*}

Identical arguments lead to
\begin{align*}
    (II) & \leq r^{-d}h^{-d}(w+1) \exp\left[-C\{(y - 2r\sigma C_{\mathpzc{k}}/c_{g, 1})  n^{1/2} w h^{d/2}/\sigma\}^{\gamma_1}\right] \\
    & \hspace{1cm} + r^{-d}h^{-d}\exp\left[-C \{(y - 2r\sigma C_{\mathpzc{k}}/c_{g, 1})n^{1/2} w h^{d/2}/\sigma\}^2/w\right].
\end{align*}

\medskip
\noindent \textbf{Step 4: Finishing.} Combining with the bounds on $E\{f(x)\}$ and $E\{\widehat{f}(x)\}$ we derived in \textbf{Step 2}, we have that, for any $r \in (0, h^{-1})$ and any $y > 2r\sigma C_{\mathpzc{k}}/c_{g, 1}$,
\begin{align*}
    & \mathbb{P} \left\{\max_{t = u+w}^{u+Q}\sup_{x \in [0, 1]^d} \frac{1}{nwh^d} \sum_{s = t+1}^{t+w} \sum_{i = 1}^n |\epsilon_{s, i}(Z_{s, i})| \mathpzc{k}\left(\frac{X_{s, i} - x}{h}\right) > y \right\} \\ 
    \leq & CQw r^{-d}h^{-d}\exp\left[-C\{(y - 2r\sigma C_{\mathpzc{k}}/c_{g, 1}) n^{1/2}   w h^{d/2}/\sigma\}^{\gamma_1}\right] \\
    & \hspace{1cm} + CQr^{-d}h^{-d}\exp\left[-C \{(y - 2r\sigma C_{\mathpzc{k}}/c_{g, 1}) n^{1/2} w h^{d/2}/\sigma\}^2/w\right]
\end{align*}
where $C > 0$ is an absolute constant.

By letting 
\[
    r \asymp \frac{\log^{1/\gamma_1} (Qwh^{-d}\gamma^{-1})}{n^{1/2} wh^{d/2}} \vee \frac{\log^{1/2} (Qwh^{-d}\gamma^{-1})}{(nw)^{1/2}    h^{d/2}}
\]
we conclude the proof.

\end{proof}

\begin{lemma}\label{lem-psi-Epsi-bound}
For any $u \in \mathbb{N}$, it holds that 
    \[
        \mathbb{P}\left[\max_{t = u+w}^{u+Q} \sup_{x \in [0, 1]^d} |\widetilde{\Psi}_{t-w, t}(x) -  E\left\{\widetilde{\Psi}_{t-w, t}(x)\right\}| > y\right]\leq c \gamma,
    \]
    where 
    \[
      y \asymp \frac{\sigma \log^{1/\gamma_1}(Qw h^{-d}\gamma^{-1})}{n^{1/2}   wh^{d/2}} \vee \frac{\sigma \log^{1/2}(Qw h^{-d}\gamma^{-1})}{(nw)^{1/2}   h^{d/2}}, \quad \gamma_1 = \frac{2 + \gamma_{\alpha}}{2\gamma_{\alpha}} 
    \]
    and $C, c > 0$ are absolute constants.
\end{lemma}

\begin{proof}[Proof of Lemma \ref{lem-psi-Epsi-bound}]
Recall that 
    \begin{align*}
        & \widetilde{\Psi}_{t-w, t}(x) = \frac{1}{nwh^p} \sum_{s = t-w+1}^t \sum_{i = 1}^n \widetilde{Y}_{s, i}\mathpzc{k}\left(\frac{X_{s, i} - x}{h}\right) \\
        = & \frac{1}{nwh^p} \sum_{s = t-w+1}^t \sum_{i = 1}^n \left\{Y_{s, i}(1)\frac{Z_{s, i}}{\pi(X_{s, i})}\mathpzc{k}\left(\frac{X_{s, i} - x}{h}\right) - Y_{s, i}(0)\frac{1 - Z_{s, i}}{1 - \pi(X_{s, i})}\mathpzc{k}\left(\frac{X_{s, i} - x}{h}\right)\right\} \\
        = & \frac{1}{nwh^p} \sum_{s = t-w+1}^t \sum_{i = 1}^n \left[\left\{\mu_{s, 1}(X_{s, i})\frac{Z_{s, i}}{\pi(X_{s, i})} - \mu_{s, 0}(X_{s, i})\frac{1 - Z_{s, i}}{1 - \pi(X_{s, i})}\right\}\mathpzc{k}\left(\frac{X_{s, i} - x}{h}\right)\right] \\
        & \hspace{1cm} + \frac{1}{nwh^p} \sum_{s = t-w+1}^t \sum_{i = 1}^n \left[\left\{\epsilon_{s, i}(1)\frac{Z_{s, i}}{\pi(X_{s, i})} - \epsilon_{s, i}(0)\frac{1 - Z_{s, i}}{1 - \pi(X_{s, i})}\right\}\mathpzc{k}\left(\frac{X_{s, i} - x}{h}\right)\right] \\
        = & (I) + (II).
    \end{align*}
    We also have that 
    \begin{align*}
        & E\left\{\widetilde{\Psi}_{t-w, t}(x)\right\} = \frac{1}{nwh^p} \sum_{s = t-w+1}^e \sum_{i = 1}^n E\left\{\widetilde{Y}_{s, i}\mathpzc{k}\left(\frac{X_{s, i} - x}{h}\right)\right\} \\
        = & \frac{1}{nwh^p} \sum_{s = t-w+1}^t \sum_{i = 1}^n  E\left[ E\left\{Y_{s, i} \frac{Z_{s, i}}{\pi(X_{s, i})} \mathpzc{k}\left(\frac{X_{s, i} - x}{h}\right) - Y_{s, i} \frac{1 - Z_{s, i}}{1 - \pi(X_{s, i})} \mathpzc{k}\left(\frac{X_{s, i} - x}{h}\right)\right\} \Bigg| X_{s, i}\right]\\
        = & \frac{1}{nwh^p} \sum_{s = t-w+1}^t \sum_{i = 1}^n  E\left\{\left[Y_{s, i}(1) - Y_{s, i}(0)\right]\mathpzc{k}\left(\frac{X_{s, i} - x}{h}\right)\right\} \\
        = & \frac{1}{nwh^p} \sum_{s = t-w+1}^t \sum_{i = 1}^n  E\left[\left\{\mu_{s, 1}(X_{s, i}) - \mu_{s, 0}(X_{s, i})\right\}\mathpzc{k}\left(\frac{X_{s, i} - x}{h}\right)\right] \\
        & \hspace{1cm} + \frac{1}{nwh^p} \sum_{s = t-w+1}^t \sum_{i = 1}^n E\left[\left\{\epsilon_{s, i}(1) - \epsilon_{s, i}(0)\right\}\mathpzc{k}\left(\frac{X_{s, i} - x}{h}\right)\right] \\
        = & (III) + (IV).
    \end{align*}
    We therefore have that 
    \begin{align*}
        \left|\widetilde{\Psi}_{t-w, t}(x) - E\left\{\widetilde{\Psi}_{t-w, t}(x)\right\}\right| \leq |(I) - (III)| + |(II) - (IV)| = (V) + (VI).
    \end{align*}
    
\medskip
\noindent \textbf{Term $(VI)$.}  Note that term $(VI)$ is a weighted average of 
    \begin{align*}
        Q_{s, i}(x) & = \left\{\epsilon_{s, i}(1)\frac{Z_{s, i}}{\pi(X_{s, i})} - \epsilon_{s, i}(0)\frac{1 - Z_{s, i}}{1 - \pi(X_{s, i})}\right\} \mathpzc{k}\left(\frac{X_{s, i} - x}{h}\right) \\
        & \hspace{1cm} - E\left[\left\{\epsilon_{s, i}(1) - \epsilon_{s, i}(0)\right\} \mathpzc{k}\left(\frac{X_{s, i} - x}{h}\right)\right].
    \end{align*}
    Due to the sub-Gaussianity of the error functions specified in Assumption \ref{assume-model-1} and the boundedness condition on the propensity scores specified in Assumption \ref{assump-causal}, the rest therefore follows from identical arguments as those in the proof of Lemma \ref{lem-abs-epsilon-bound}.

\medskip
\noindent \textbf{Term $(V)$.} Note that term $(V)$ is a weighted average of 
    \begin{align*}
        R_{t, i}(x) & = \left\{\mu_{s, 1}(X_{s, i})\frac{Z_{s, i}}{\pi(X_{s, i})} - \mu_{t, 0}(X_{s, i})\frac{1 - Z_{s, i}}{1 - \pi(X_{s, i})}\right\} \mathpzc{k}\left(\frac{X_{s, i} - x}{h}\right) \\
        & \hspace{1cm} -  E\left[\left\{\mu_{s, 1}(X_{s, i}) - \mu_{s, 0}(X_{s, i})\right\} \mathpzc{k}\left(\frac{X_{s, i} - x}{h}\right)\right].
    \end{align*}
    Due to the boundedness conditions on the mean functions specified in Assumption \ref{assume-model-1} and the boundedness condition on the propensity scores specified in Assumption \ref{assump-causal}, the rest therefore follows from identical arguments as those in the proof of Lemma \ref{lem-abs-epsilon-bound}.
\end{proof}

\begin{lemma}\label{lem-ghat-g-diff}
For any $u \in  N$, it holds that 
    \[
        \mathbb{P}\left\{\max_{t = u+w}^{u+Q} \sup_{x \in [0, 1]^d} |\widehat{g}_{t-w, t}(x) - g(x)| > y\right\} \leq c \gamma,
    \]
    where 
    \[
      y \asymp \frac{\sigma \log^{1/\gamma_1}(Qw h^{-d}\gamma^{-1})}{ n^{1/2}   wh^{d/2}} \vee \frac{\sigma \log^{1/2}(Qw h^{-d}\gamma^{-1})}{(nw)^{1/2} h^{d/2}} + h, \quad \gamma_1 = \frac{2 + \gamma_{\alpha}}{2\gamma_{\alpha}} 
    \]
    and $C, c > 0$ are absolute constants.
\end{lemma}

\begin{proof}[Proof of Lemma \ref{lem-ghat-g-diff}]
Note that
    \begin{align*}
        |\widehat{g}_{t-w, t}(x) - g(x)| \leq |\widehat{g}_{t-w, t}(x) - E\{\widehat{g}_{t-w, t}(x)\}| + |E\{\widehat{g}_{t-w, t}(x)\} - g(x)| = (I) + (II).
    \end{align*}
    
\medskip
\noindent \textbf{Term $(I)$.}  It holds that
    \begin{align*}
        (I) = \left|\frac{1}{nwh^d} \sum_{s = t-w+1}^t \sum_{i = 1}^n \left[\mathpzc{k}\left(\frac{X_{s, i} - x}{h}\right) -E\left\{\mathpzc{k}\left(\frac{X_{s, i} - x}{h}\right)\right\}\right]\right|.
    \end{align*}
    Due to the boundedness condition on the kernel functions specified in Assumption \ref{assump-kernel}, the rest therefore follows from identical arguments as those in the proof of Lemma \ref{lem-abs-epsilon-bound}.
    
\medskip
\noindent \textbf{Term $(II)$.} It holds that
    \begin{align*}
        (II) & = \left|h^{-d}\int_{\mathcal{X}} \mathpzc{k}\left(\frac{u - x}{h}\right)g(u)\dint u - g(x)\right| = \left|\int_{\mathcal{X}}\mathpzc{k}(y)g(x + hy)\dint y - g(x)\right| \\
        & \leq \int_{\mathcal{X}} |g(x + hy) - g(x)| \mathpzc{k}(y) \dint y \leq hC_{\mathrm{Lip}} \int_{\mathcal{X}} \|y\|\mathpzc{k}(y) \dint y \leq Ch,
    \end{align*}
    where we abuse notation with a generic absolute constant $C > 0$.

\end{proof}

\begin{lemma}\label{lem-epis-tau-diff-bound}
For any $u \in  N$, it holds that 
    \[
        \mathbb{P}\left\{\max_{t = u+w}^{u+Q} \sup_{x \in [0, 1]^d} \left|\frac{ E\left\{\widetilde{\Psi}_{t-w, t}(x)\right\}}{\widehat{g}_{t-w, t}(x)} - \tau_{t-w, t}(x)\right| > y\right\} \leq c \gamma,
    \]
    where 
    \[
      y \asymp \frac{\sigma \log^{1/\gamma_1}(Qw h^{-d}\gamma^{-1})}{ n^{1/2}   wh^{d/2}} \vee \frac{\sigma \log^{1/2}(Qw h^{-d}\gamma^{-1})}{(nw)^{1/2}  h^{d/2}} + h, \quad \gamma_1 = \frac{2 + \gamma_{\alpha}}{2\gamma_{\alpha}} 
    \]
    and $C, c > 0$ are absolute constants.
\end{lemma}

\begin{proof}[Proof of Lemma \ref{lem-epis-tau-diff-bound}]
It holds that
    \begin{align*}
        & E\left\{\widetilde{\Psi}_{t-w, t}(x)\right\} = \frac{1}{nwh^p} \sum_{s = t-w+1}^t \sum_{i = 1}^n  E\left\{\widetilde{Y}_{s, i} \mathpzc{k}\left(\frac{X_{s, i} - x}{h}\right)\right\} \\
        = & \frac{1}{nwh^p} \sum_{s = t-w+1}^t \sum_{i = 1}^n E\left[\left\{Y_{s, i}(1) \frac{Z_{s, i}}{\pi(X_{s, i})} - Y_{s, i}(0) \frac{1 - Z_{s, i}}{1 - \pi(X_{s, i})}\right\} \mathpzc{k}\left(\frac{X_{s, i} - x}{h}\right)\right] \\
        = & \frac{1}{nwh^p} \sum_{s = t-w+1}^t \sum_{i = 1}^n  E\left( E\left[\left\{Y_{s, i}(1) \frac{Z_{s, i}}{\pi(X_{s, i})} - Y_{s, i}(0) \frac{1 - Z_{s, i}}{1 - \pi(X_{s, i})}\right\} \mathpzc{k}\left(\frac{X_{s, i} - x}{h}\right)\Bigg| X_{s, i}\right] \right) \\
        = & \frac{1}{nwh^p} \sum_{s = t-w+1}^t \sum_{i = 1}^n  E\left(\left[\left\{Y_{s, i}(1) \frac{Z_{s, i}}{\pi(X_{s, i})} - Y_{s, i}(0) \frac{1 - Z_{s, i}}{1 - \pi(X_{s, i})}\right\} \Bigg|X_{s, i}\right] \mathpzc{k}\left(\frac{X_{s, i} - x}{h}\right)\right) \\
        = & \frac{1}{nwh^p} \sum_{s = t-w+1}^t \sum_{i = 1}^n E\left[E\left\{Y_{s, i}(1) - Y_{s, i}(0)\Big|X_{s, i}\right\}\mathpzc{k}\left(\frac{X_{s, i} - x}{h}\right)\right] \\
        = & \frac{1}{nwh^p} \sum_{s = t-w+1}^t \sum_{i = 1}^n \int_{\mathcal{X}} \tau_s(u) \mathpzc{k}\left(\frac{u-x}{h}\right) g(u) \dint u. 
    \end{align*}    
    We then have
    \begin{align*}
        & \left|\frac{E\left\{\widetilde{\Psi}_{t-w, t}(x)\right\}}{\widehat{g}_{t-w, t}(x)} - \tau_{t-w, t}(x)\right| = \left|\frac{\frac{1}{nwh^p} \sum_{s = t-w+1}^t \sum_{i = 1}^n \int_{\mathcal{X}} \tau_s(u) \mathpzc{k}\left(\frac{u-x}{h}\right) g(u) \dint u}{\widehat{g}_{t-w, t}(x)} - \tau_{t-w, t}(x)\right| \\
        \leq & \left|\frac{\frac{1}{nwh^p} \sum_{s = t-w+1}^t \sum_{i = 1}^n \int_{\mathcal{X}} \tau_s(u) \mathpzc{k}\left(\frac{u-x}{h}\right) g(u) \dint u}{\widehat{g}_{t-w, t}(x)} - \tau_{t-w, t}(x)\frac{g(x)}{\widehat{g}_{t-w, t}(x)}\right| \\
        & \hspace{1cm} + \left|\tau_{t-w, t}(x)\frac{g(x)}{\widehat{g}_{t-w, t}(x)} - \tau_{t-w, t}(x)\right| \\
        = & (I) + (II).
    \end{align*}
    
\medskip
\noindent \textbf{Term (I)}.  Note that
    \begin{align*}
        (I) & = \left|\frac{h^{-p}\int_{\mathcal{X}} \tau_s(u) \mathpzc{k}\left(\frac{u - x}{h}\right)g(u)\dint u - \tau_s(x)g(x)}{\widehat{g}_{t-w, t}(x)}\right| \\
        & \leq \left|h^{-p}\int_{\mathcal{X}} \tau_s(u) \mathpzc{k}\left(\frac{u - x}{h}\right)g(u)\dint u - \tau_s(x)g(x)\right| \frac{1}{g(x) - \|\widehat{g}_{t-w, t} - g\|_{\infty}} \\
        & = \left|\int_{\mathcal{X}} \left\{\tau_s(x + hy) g(x +hy) - \tau_s(x)g(x)\right\}\mathpzc{k}(y) \dint y \right| \frac{1}{g(x) - \|\widehat{g}_{t-w, t} - g\|_{\infty}} \\
        & \leq \frac{C_{\mathpzc{k}}}{c_{g, 1} - |\widehat{g}_{t-w, t} - g\|_{\infty}}\int_{\mathcal{X}}|\tau_s(x + hy) g(x +hy) - \tau_s(x + hy)g(x)| \mathpzc{k}(y) \dint y \\
        & \hspace{1cm} + \frac{C_{\mathpzc{k}}}{c_{g, 1} - |\widehat{g}_{t-w, t} - g\|_{\infty}}\int_{\mathcal{X}}|\tau_s(x + hy) g(x) - \tau_s(x)g(x)| \mathpzc{k}(y) \dint y \\
        & \leq \frac{C_{\mathpzc{k}} C_{\mathrm{Lip}} h}{c_{g, 1} - \|\widehat{g}_{t-w, t} - g\|_{\infty}} \left\{\int_{\mathcal{X}} |\tau_s(x + hy)| \|y\| \mathpzc{k}(y) \dint y  + |g(x)|\int_{\mathcal{X}} \|y\| \mathpzc{k}(y)\dint y\right\} \\
        & \leq Ch,
    \end{align*}
    where the last inequality holds provided that 
    \[
        \|\widehat{g}_{t-w, t} - g\|_{\infty} \leq c_{g,1}/2.
    \]
    
\medskip
\noindent \textbf{Term (II)}.  Note that
    \begin{align*}
        (II) \leq \frac{\|\tau_{t-w, t}\|_{\infty}}{\widehat{g}_{t-w, t}(x)} \|\widehat{g}_{t-w, t} - g\|_{\infty} \leq \frac{\|\tau_{t-w, t}\|_{\infty}}{c_{g, 1} - \|\widehat{g}_{t-w, t} - g\|_{\infty}} \|\widehat{g}_{t-w, t} - g\|_{\infty} \leq C \|\widehat{g}_{s, e} - g\|_{\infty}.
    \end{align*}
    
The proof therefore concludes due to the results in Lemma \ref{lem-ghat-g-diff}.     
\end{proof}

\section{Additional experiments} \label{sec:exp2}

We conduct the experiments from Section \ref{sec-numerical}  with $w=7$, to show the robustness of Algorithm \ref{alg-main} with respect to the window width $w$.  The results in Table \ref{tab2} once again validate the findings from Section \ref{sec-numerical}.

\begin{table}[t!]
	\centering
	\caption{ 	\label{tab2}  Delay averaging over 50 Monte Carlo simulations for different scenarios  and choices of tuning parameters. }
	\medskip
	\setlength{\tabcolsep}{3.5pt}
	\begin{small}
		\begin{tabular}{rr|cc|cc||cc|cc} 
			\hline
			&           &               \multicolumn{4}{c}{Scenario 1}                &    \multicolumn{4}{c}{Scenario 2}\\ 
			$d$  & $h$      & Delay                                     & Delay                                      & Delay                                                  & Delay                  & Delay                                     & Delay                                      & Delay                                                  & Delay  \\ 
			&                  &          ($\Gamma =20$)     &     ($\Gamma =20$)          &          ($\Gamma =40$)               &($\Gamma =40$) &          ($\Gamma =20$)     &     ($\Gamma =20$)          &          ($\Gamma =40$)               &($\Gamma =40$) \\
			&              &                 Alg 1                 &              DK                                   &        Alg 1                                             &  DK                       &                 Alg 1                          &       DK                                  &                         Alg 1                 &  DK\\
			\hline    
			3     &    20         &       13.4 (19.4)      &            15.0 (18.3)                       &      28.5 (24.5)                              &   27.2 (30.0)&              4.1 (4.9)                         &           5.9 (15.0)        &           14.7 (24.0)                          &15.8 (27.3)\\
			3     &    4          &    14.6 (14.2)          &            21.3 (26.2)                         &           32.5 (27.8)                         &    33.0 (30.2)        &       16.6 (21.9)                &           12.5 (19.8)               &        27.1 (28.5)              & 35.1 (32.7) \\
			6      &    20         &        16.0 (17.6)          &            8.4 (15.6)               &              29.3 (25.9)                    &          21.7 (28.6)       &              24.4 (26.8)               &                  17.1 (26.5)                 &              36.1 (26.9)                  & 32.7 (32.3)\\
			6      &    4          &         15.7    (15.3)         &       11.5 (16.3)            &             31.4 (25.8)                         &       20.9 (24.5)               &      16.9 (24.3)                &         20.1 (29.6)    & 30.2 (30.1)& 26.9 (32.2)\\			          			          			          
			\hline
			\hline
			&           &               \multicolumn{4}{c}{Scenario 3}                &    \multicolumn{4}{c}{Scenario 4}\\ 
			$d$  & $h$      & Delay                                     & Delay                                      & Delay                                                  & Delay                  & Delay                                     & Delay                                      & Delay                                                  & Delay  \\ 
			&                  &          ($\Gamma =20$)     &     ($\Gamma =20$)          &          ($\Gamma =40$)               &($\Gamma =40$) &          ($\Gamma =20$)     &     ($\Gamma =20$)          &          ($\Gamma =40$)               &($\Gamma =40$) \\
			&                             &                 Alg 1                 &              DK                                   &        Alg 1                                             &  DK                       &                 Alg 1                          &       DK                                  &                         Alg 1                 &  DK\\
			\hline
			3    &    20         & 9.4 (17.7)                    &       16.0 (27.2)                             &          22.7 (29.5)                            &       26.2 (32.3)          &          6.2 (14.8)      &             22.8 (31.1)       &         22.7 (31.1)              & 29.4 (33.2)\\
			3     &    4          &    5.8 (14.9)                     &       22.8 (31.0)                          &        16.1  (27.2)                                &  39.3 (33.7)            &    12.9 (24.2)        &         12.9 (24.7)  &                        22.8(31.2)            &39.3 (33.5)\\       
			6      &    20         &      12.6 (24.3)             &      16.0 (27.2)                              &        25.9 (32.5)                          & 22.7 (31.0)              &         12.8 (24.2)                  &       16.2 (27.1)          &             26.0 (32.3)                 & 26.1 (32.3) \\
			6      &    4          &       12.6 (24.3 )               &                30.0 (33.8)                 &   32.0 (33.8)                                 &     32.9 (24.2)         &             9.3 (20.3)            &                9.5 (20.4)  &                          19.5 (29.4)                  &29.4 (33.6)\\				   	     
			\hline \hline			
		\end{tabular}
	\end{small}
\end{table}

\section{Estimating two mean functions vs.~estimating the CATE function}
\label{sec:ad_fig}

We illustrate with an example further explaining why estimating the CATE function (true treatments) with a single kernel estimator (One-K) can be better than using the difference of two kernel estimators (Two-K) of the two mean functions, respectively.  We generate data as
\[
     \begin{cases}
        Y_i = Y_i(1) Z_i  +   Y_i(0) (1-Z_i) \\
        Y_{i}(1) =  \mu_{0}(X_{i}) + \tau(X_{i}) +  \epsilon_{i}(1)\\
     	Y_{i}(0) =   \mu_{0}(X_{i}) + \epsilon_{i}(0)\\
      	\mathbb{P}(Z_{i}=1| X_{i}) = 0.5\\
     	\epsilon_{i}(l) \overset{\mathrm{ind}}{\sim} \mathcal{N}(0,1), \quad l=0,1, \\
        X_{i} \overset{\mathrm{ind}}{\sim} \mathrm{Unif}([0,1]),
     \end{cases} \quad  i=1,\ldots,n = 4000,
\]
$\mu_0(x) = \cos(100/x)$ and $\tau(x) = \{1 + \exp(-20(x-1/3))\}^{-1}$.

To estimate $\tau(\cdot)$, we consider the  One-K estimator defined as
\[
\hat{\tau}^1(x) = \frac{\sum_{i = 1}^n \left\{Y_{ i} \left(\frac{Z_{i}}{\hat{p}} - \frac{1 - Z_{i}}{1 - \hat{p}}\right)\right\} \mathpzc{k}\left(\frac{X_{i} - x}{h}\right)}{\sum_{i = 1}^n \mathpzc{k}\left(\frac{X_{i} - x}{h}\right)}
\]
with $\hat{p} = \sum_{i=1}^nZ_i/n$ and with $h >0$ a bandwidth.  We also consider the Two-K estimator defined as 
\[
\hat{\tau}^2(x) = \frac{\sum_{i = 1}^n  Y_{ i}Z_i \mathpzc{k}\left(\frac{X_{i} - x}{h_1}\right)}{\sum_{i = 1}^n  Z_i\mathpzc{k}\left(\frac{X_{i} - x}{h_1}\right)}\,-\, \frac{\sum_{i = 1}^n  Y_{ i}(1-Z_i )\mathpzc{k}\left(\frac{X_{i} - x}{h_2}\right)}{\sum_{i = 1}^n  (1-Z_i)\mathpzc{k}\left(\frac{X_{i} - x}{h_2}\right)},
\]
where $h_1, h_2>0$ are bandwidth parameters. Selecting $h$, $h_1$ and $h_2$ with the plug-in method from \cite{ruppert1995effective}, we obtain the results in Figure \ref{fig4} showing that One-K clearly outperforms the method Two-K, which in turn demonstrates the superiority of Algorithm \ref{alg-main} over the DK method considered in Section \ref{sec-numerical}.

\end{document}